\pdfoutput=1
\documentclass[conference,10pt,letterpaper]{IEEEtran}
\usepackage{color, url, verbatim}
\usepackage{algorithm,algorithmic,booktabs}
\usepackage{graphicx, bbm, latexsym, amsfonts, verbatim, amsmath, amssymb}
\newcommand{\zblue}{\textcolor{black}}
\newcommand{\tred}{\textcolor{black}}
\newtheorem{observation}{Observation}
\newtheorem{define}{Definition}
\newtheorem{definition}{Definition}
\newtheorem{assumption}{Assumption}
\newtheorem{proposition}{Proposition}
\newtheorem{theorem}{Theorem}
\newtheorem{lemma}{Lemma}
\newtheorem{proof}{Proof}

\def\Pe{{\mathcal P}}
\newcommand{\nproof}{\noindent{\em Proof: }}
\newcommand{\proofend}{\hfill$\Box$ \newline\noindent}

\title{\zblue{Robust Network Function Virtualization}}

\author{
\IEEEauthorblockN{Tachun~Lin}
\IEEEauthorblockA{Department of Computer Science and Information Systems\\
Bradley University, Peoria, IL USA}
Email: djlin@bradley.edu
\and
\IEEEauthorblockN{Zhili~Zhou}
\IEEEauthorblockA{United Airlines\\
Chicago, IL USA}
Email: zhili.zhou@united.com
}

\begin{document}
\maketitle

\begin{abstract}
\tred{Network function virtualization (NFV) enables on-demand network function (NF) deployment providing agile and dynamic network services. Through an evaluation metric that quantifies the minimal reliability among all NFs for all demands, service providers and operators may better facilitate flexible NF service recovery and migration, thus offer higher service reliability.}
\zblue{In this paper, we present evaluation metrics on NFV reliability and solution approaches to solve robust NFV under random NF-enabled node failure(s). 
We demonstrate how to construct an auxiliary NF-enabled network and its mapping onto the physical substrate network. With constructed NF-enabled network, we develop pseudo-polynomial algorithms to solve the robust NF and SFC $s-t$ path problems -- subproblems of robust NFV. We also present approximation algorithms for robust NFV with the SFC-Fork as the NF forwarding graph. Furthermore, we propose exact solution approaches via mixed-integer linear programming (MILP) under the general setting. Computational results show that our proposed solution approaches are capable of managing robust NFV in a large-size network.}

\end{abstract}

\section{Introduction}\label{sec:introduction}
The development of 5G networks targets to deliver ultra-reliable and super low latency communication~\cite{alliance20155g,osseiran2014scenarios}, which supports dynamic requests over large-scale cross-domain networks. Through network function virtualization (NFV), a 5G-enabling technique, network functions (NFs) are decoupled from costly proprietary networking hardware and are realized through their software implementation of virtual network functions (VNFs) running on industry-standard commercial off-the-shelf (COTS) hardware~\cite{abdelwahab2016network,han2015network}. Radio signal processing~\cite{peng2015fronthaul} and mobile/optical networking~\cite{mijumbi2016network} are also applying the NFV and deploy VNFs on NF-enabled physical infrastructures, such as virtual machines and containers, and provide on-demand NF services~\cite{etsi20175G}.

NF service providers provision, manage, and orchestrate VNFs with NFV management and orchestration architectures (MANO)~\cite{etsi15nfvi} for end users which request a sequence of NFs called ``\textit{service function chaining}'' (SFC). An instance of SFC is (firewall $\rightarrow$ intrusion prevention system $\rightarrow$ load balancer). We let ``\textit{non-chained}'' NF requests denote the NF requests without a specified sequence. With NFV, network operators allocate and reallocate VNF instances and route network traffic between service functions. Hence, NFV does not only provide more flexibility but also shortens the enabling time of new NF services~\cite{etsi15nfvi}. To realize end-user demands with NF requests, the NF provisioning problem, which determines the physical infrastructure to deploy VNF instances to fulfill NF requests, arises. We illustrate an instance of NF provisioning problem in Fig.~\ref{fig:nfDply}.
\begin{figure}[!h]
\centering
\includegraphics[scale=0.5]{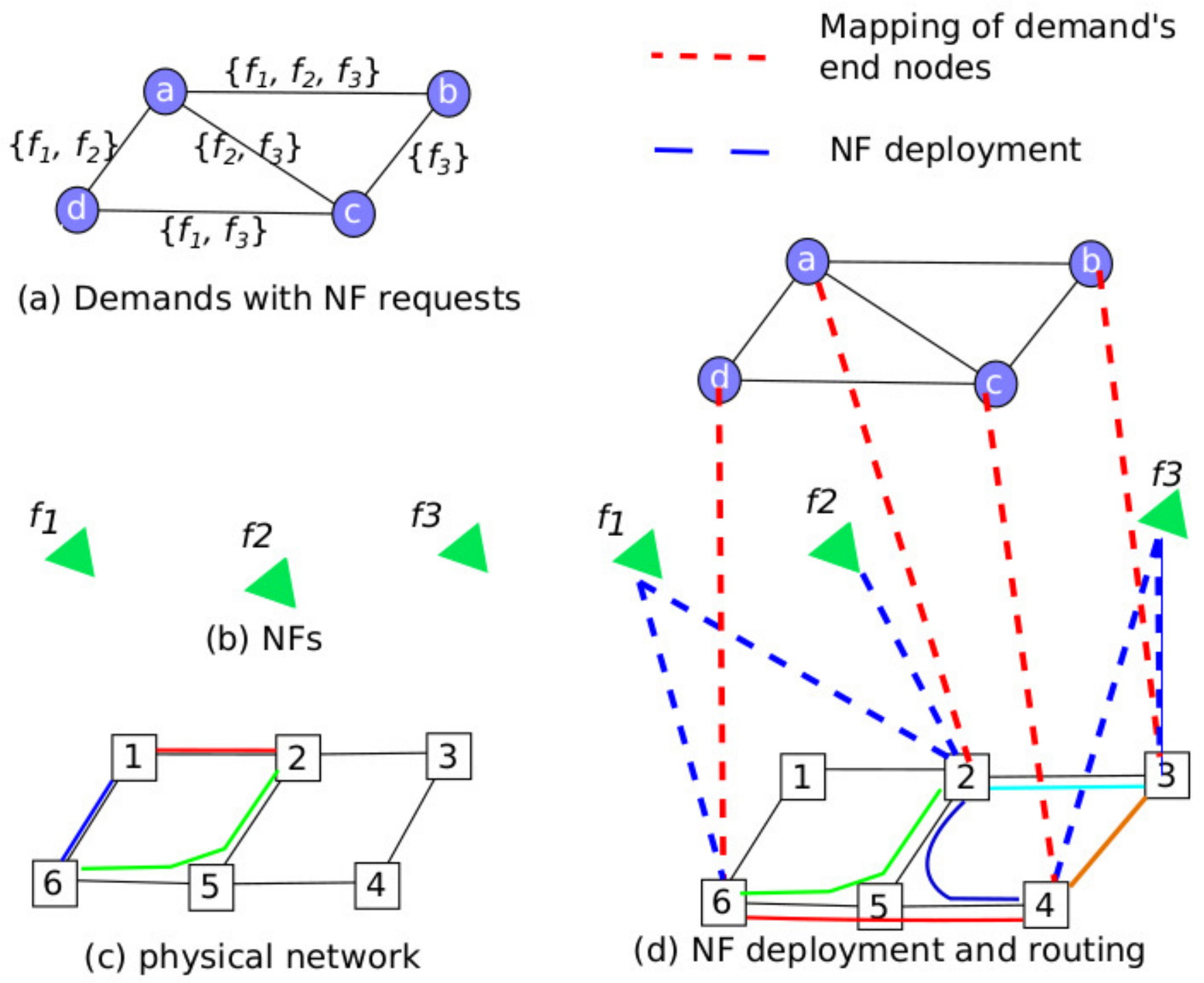}
\caption{NF provisioning with NFs deployment on NF-enabled nodes}\label{fig:nfDply}
\end{figure}

NFV MANO also supports NF recovery and migration, the major approaches to guarantee continuity, resilience, and security of NF services~\cite{eramo2017approach,mijumbi2016network,wang2017consistent}. When a VNF instance is not reachable~\cite{etsi2014nfvMO}, MANO initiates the fail-over to other available NF instances and automatically recovers NF services, or instantiates new VNFs~\cite{xia2014VPool}. Meanwhile, dynamic and flexible VNF migration also reduces power consumption and the burden on hardware capacity~\cite{mijumbi2016network}. To support NF recovery and migration, reserving physical resources should also be considered in the NF provisioning problem.

While the above studies provide valuable insights from different aspects of NF services, they cannot be used to quantify system capability and reliability to support NF recovery and migration~\cite{eramo2017migration,hawilo2017orchestrating,han2017resiliency} as well as seamless NFV state transitions~\cite{nobach2017statelet,wang2016transparent} under component failure(s). Motivated by the objective of providing ultra-reliable 5G services, we \tred{studied} the robust NF provision problem \tred{in~\cite{RNDM18}}
which takes into account the uncertain failures of NF-enabled nodes from the network operators' perspective. \tred{Robust evaluation metrics proposed in~\cite{RNDM18} on robust VNF provisioning target to} provide VNF managers/orchestrators a way to evaluate the strategies to instantiate VNFs on available NF-enabled nodes (NF resource pools) based on the information of the physical infrastructure and resource utilization. 

\zblue{This study addresses robust NF provisioning and related network design and routing problems. The robust NF provisioning problem determines the location of VNF instance deployment, and possible NF request fulfillment is determined with a robust evaluation metric as the objective to handle random NF-enabled node failures. Extended from ~\cite{RNDM18}, we study three sets of network design and routing problems for both considering and not considering NF-enabled node failure: (1) the minimal weighted SFC $s-t$ path problem, which determines the minimal weighted path for SFC requests visiting all required NFs in sequences through NF-enabled nodes; (2) VNF provisioning with SFC-Fork as the forwarding graph structure, which determines deployment locations of VNF instances realizing all SFC requests; and (3) VNF provisioning with general NF forwarding graph.}
The first set of problems is a fundamental problem in NFV which helps establish the end-to-end route to fulfill NF requests. The second set has SFC-Fork as the NF forwarding graph, which is the commonly established forwarding graph in NFV 5G implementation (see in~\cite{cheng2015enabling, jalalitabar2018service, nfv2014001, openSFC}). To address these problems, we construct an auxiliary NF-enabled network that serves as an intermediate layer between the NF forwarding graph and the physical substrate network and provides all possible connections among NF-enabled nodes. We present pseudo-polynomial algorithms for the NF and SFC $s-t$ path problems and approximation algorithms for the NF provisioning problem with SFC-Fork. We also validate our proposed solution approaches with computational results over small and large scale national-wide physical networks.

\tred{We highlight our contributions in this paper as follows.}

\begin{enumerate} 
\item We propose robust evaluation metrics~\cite{RNDM18} on robust VNF provisioning to provide VNF managers/orchestrators a way to evaluate strategies in instantiating VNFs on available NF-enabled nodes (in NF resource pools) based on the physical infrastructure and resource utilization.
\item We construct multi-layer graphs and establish their corresponding mapping relationships, which provide network structures to solve the NFV design problems. 
\item We provide pseudo-polynomial algorithms for NF and SFC $s-t$ path problems in both deterministic and robust settings.
\item We demonstrate the existence of bi-factor approximation algorithms on NF provisioning with SFC-Fork and propose corresponding algorithms. 
\item We propose a two-step parameterized path reduction technique in approximation algorithm design to manage branching structures in tree networks.
\item \tred{With the insight obtained from the approximation algorithm for the optimal NF provisioning problem, we develop an approximation algorithm for the robust NF provisioning problem.}
\end{enumerate}
\zblue{In short, the robust NF provisioning problems address the problem of sequential location selection/resource allocation and the routing through selected locations, which add new variants and dimensions in the traditional location and location-routing problems. Our proposed evaluation metrics and solution approaches serve the purpose of dealing with these new variants.}

The rest of the paper is organized as follows. We first present the related works, problems, and their solution approaches in section~\ref{sec:literature}. \tred{Especially, we summarize the approximation algorithms for $k$-level facility location problems and their robust/fault-tolerant relatives.} In section~\ref{sec:prDescription}, we review the evaluation metric for NF services, the robust NF-service evaluation metric, for non-chained NFs and SFCs and define the robust NF provisioning problem and related subproblems. We present corresponding solution approaches in Section~\ref{sec:formulation}, where an auxiliary NF-enabled network is constructed. We introduce the (robust) SFC path algorithms and the mixed-integer programming formulations to solve robust NF provisioning for non-chained NF and SFC requests, respectively. The experiment setting and computational results to validate our proposed approaches are given in section~\ref{sec:computation}. We also demonstrate the low bound benchmark of the robust NF-service evaluation metrics, followed by the conclusions and future research directions in section~\ref{sec:conclusion}. \tred{We also include in the appendix a bit theoretical study to entertain the audience with the fondness of theorems, proofs, and approximation ratio of the algorithm.}

\section{Literature Review}\label{sec:literature}
\subsection{Network Function Virtualization Techniques}\label{subsec:nfv}
We review in this section the works on NFV resource allocations and existing solution approaches and focus on the approximation algorithms for facility location problems and relevant location routing problems. Related topics have been reviewed in survey papers, such as NFV architectures~\cite{bonfim2018integrated}, mobile applications~\cite{nguyen2017sdn}, and NF deployment/provisioning related resource allocation~\cite{nguyen2017resource,xie2016service}. Most recent works on NFV MANO focus on NFV instantiation~\cite{carpio2017vnf,zhang2018theory}, orchestration~\cite{zeng2016orchestrating}, management~\cite{d2017game}, and scheduling~\cite{alameddine2017interplay}. Newly developed technologies are capable to support NFV in various telecommunication systems.

\subsection{Related Problems and Solution Approaches}\label{subsec:relatedProblem}
Most resources allocation problems for NF deployment/provisioning~\cite{nguyen2017resource} are under the setting that given a physical substrate network (an available NFVI), the set of NF-enabled nodes is a subset of physical nodes and the end-to-end NF demands are established/realized through paths in the physical substrate network. NF deployment problem determines locations and copies of VNF instances deployed and generates end-to-end routes (static/dynamic) for NF requests. On top of the NF deployment problem, NF provisioning problem further estimates physical resources considering also the quality-of-service (QoS). Hence, NF deployment and provisioning problems belong to location-routing problems.

SFC route generation~\cite{sallam2018shortest} and NF forwarding graphs embedding~\cite{gupta2018scalable} are the two approaches to manage SFC requests either individually or jointly. Hence, even with the simplest case, the minimal weighted SFC path problem is different from the shortest path problem, which requires visiting VNF instantiated physical nodes in the desired sequence. \cite{cohen2015near} decomposes non-chained VNF deployment into two stages: (1) VNF instance deployment via facility location problem, and (2) VNF instance assignment for NF service requests via assignment problem, which allows NF service requests to be fulfilled by splittable flows. They provide a mixed integer linear programming (MILP) formulation for each stage with the objective to minimize the system-wide operation costs and build an upper-bounded heuristic algorithm. \cite{bari2015orchestrating} studies SFC deployment and presents a MILP formulation which realizes SFC requests via multiple virtual network embeddings~\cite{chowdhury2009virtual}. A heuristic algorithm based on a multi-stage directed graph and Viterbi algorithm is proposed, which takes each SFC as a virtual network and maps each virtual network onto the given physical substrate network.
~\cite{rost2016service} studies multiple SFCs (forwarding graph) embedding and proposes a polynomial-time approximation algorithm based on random routing techniques with linear programming relaxation.~\cite{even2016approximation} provides a randomized approximation algorithm leveraging multi-commodity flows for path computation and function placement.~\cite{sallam2018shortest} focuses on the SFC counterpart of standard graph theory problems, such as the minimal weighted SFC path and the SFC maximal flow. The proposed pseudo-polynomial algorithms solve the minimal weighted SFC paths on a transferred network graph and a special case of SFC maximal flow problem.

With the capacity limitation on the physical substrate network and single-path realization of end-user demands, \cite{lin2016demand} deploys VNF instances in an optical backbone network and formulates the problem as unit flow multicommodity flow problem. The resource competition between NF instantiation and demand routing is captured using game theory setting, which reveals that a system-wide optimal solution may not be preferable for network service providers and individual end-users.~\cite{d2017game} takes into account the selfishness and competitiveness of end-users' behavior and formulates an atomic weighted congestion game for SFC routing. It proposes a polynomial-time algorithm to achieve Nash equilibrium with a bounded price of anarchy. Besides planning and operation costs, QoS is another key evaluation metric for network services, especially in a 5G environment.
One of the 5G-PPP research projects funded by the European Union, 5G-NORMA~\cite{5gNormaQoS}, provides QoS requirements on NF service during 5G implementations considering deployment failure due to lack of infrastructure resources.~\cite{fan2015grep} demonstrates that controlled redundancy provides extra protection and recovery capability for network services when a physical failure occurs. It develops an online heuristic algorithm minimizing physical resource consumption and guaranteeing service reliability with backup resources.~\cite{qu2016reliability} studies the reliability-aware NF provisioning problem and proposes MILP formulation and greedy based heuristic algorithm, in which extra backup VNF instances for NF requests are deployed.
From the operators' perspective on managing the QoS of NF services, in this study, we fill in the gap with the study of the QoS-controlled NF provisioning problem. We first propose a new evaluation metric on NF-service reliability under the worst-case scenario, followed by studying \textit{robust NF provisioning} under the failure uncertainty of NF-enabled node. As mentioned earlier, two subproblems are involved in the NF provisioning problem: (1) The minimal-weighted SFC path problem: different from~\cite{sallam2018shortest}, we first propose a Dijkstra-like algorithm with an extended level of information; and (2) \zblue{robust NF provisioning: we present mixed-integer programming formulations as its exact solution approach}.

\subsubsection{Facility Location Approximations}\label{subsubsec:approximation}
\tred{We briefly summarize some existing approximation algorithms for the $k$-level facility location problem and robust fault-tolerant facility location problem with a restriction on the standard uncapacitated facility location problem where no penalty is allowed. These related works would serve as the foundation for us to discuss the potential of the approximation algorithms for robust NF provisioning.
The $k$-level facility location problem, which has a client set and $k$ types/levels of facility sets, targets to connect clients to opened facilities at each level (in the order of level 1 to level $k$) with minimal costs/weights.
Through linear programming (LP) relaxation, ~\cite{guha1999greedy} generalized the 1-level uncapacitated facility location problem and developed an algorithm with an approximation ratio of 1.463.~\cite{krishnaswamy2012inapproximability} improved the ratio to 1.61 for general $k$. Harder than the 1-level facility location problem, the $k$-level facility location problem has a currently best-known approximation with a ratio of 1.488~\cite{li20131}. The constant factor approximation was started from a 2-level problem ($k=2$) with a 3.16-approximation algorithm~\cite{shmoys1997approximation}, and a 3-approximation for general $k$~\cite{aardal19993}.~\cite{ageev2004improved} demonstrated that the $k$-level problem could be reduced to $(k-1)$-level problem and a 1-level problem and provided a 2.43-approximation, which was further improved by~\cite{zhang2006approximating} and \cite{byrka2007optimal}\cite{byrka2016improved} through LP primal-dual, randomization, and facility scaling techniques.~\cite{chechik2009robust} introduced the robust fault-tolerance problem, which has a two-stage robust optimization setting with the first stage determining the facility location and client assignment, and the second stage reconnecting clients if up to a total of $\alpha$ connected facility failed (which were opened in the first stage). Based on such setting,~\cite{chechik2009robust} developed a (7.5$\alpha$+1.5)-approximation algorithm, which was further improved to a $(k+5+k/4)$-approximation in~\cite{byrka2010lp} through LP-rounding.}

\section{Notations and Problem Description}\label{sec:prDescription}
In this section, we first provide the general notations used in the discussions. We then propose the robust NF-service evaluation metric and the robust VNF provisioning problem to minimize the number of instantiated VNFs while maximizing the robust NF-service evaluation metric.
Let $G_P=(V_P, A_P)/G_P=(V_P, E_P)$ be the physical infrastructure with node set $V_P$ and arc/edge set $A_P$/$E_P$, required NF set $F$, and end-to-end service request $\mathcal{D}=\{d_{st}\}$. Let node set $V^{f}_{P}\in V_P$ denote a physical resource pool for NF $f$ (candidate physical nodes to deploy $f$) and  $V^{F}_{P} = \cup_{f\in F}V^{f}_{P}$ be the NF-enabled node set. 
Each of the NF-enabled nodes is with failure probability $\rho_{i}$, $i\in V^{F}_{P}$ and $0\leq \rho_i\leq 1$.

We assume that the NF requests $d_{st}, s, t\in V_L$ are known a priori. 
Let $d_{st}$ be a tuple $[(s,t), \sigma_{st}, F_{st}]$, where $\sigma_{st}$ indicates whether the request is with SFC or not; if yes, $\sigma_{st}=1$, otherwise $\sigma_{st}=0$.
\zblue{We let $\widetilde{P}$ and $\overrightarrow{P}$ be the undirected and directed path sets in the physical network. A demand $(s, t)$ with NF requests is \emph{fulfilled} if it is routed through path $p_{st}\in \overrightarrow{P}$ visiting all required NFs  in the sequence defined in SFC when $\sigma_{st}=1$, or otherwise routed through undirected path $\eta_{st}\in \widetilde{P}$ visiting all required NFs with $\sigma_{st}=0$. To simplify the notation, we let $P$ represent the path set containing all undirected and directed paths of all NF requests.}

\tred{Without loss of generality, we assume that the physical substrate network is at least 2-connected in this paper.}
Notations and parameters are summarized in Table~\ref{tbl:notation}.

\begin{table}[h]
\caption{Notations and parameters}\label{tbl:notation}
\begin{tabular}{p{1.5cm}p{6.5cm}}
{Notation}  & {Description}\\
\hline
$G_P(V_P, A_P)$, $ G_P(V_P, E_P)$   & Physical infrastructure $G_P$ with $V_P, A_P/E_P$ as its node and arc/link sets, respectively\\
$G_L(V_L, E_L)$      & Logical network $G_L$ with $V_L, E_L$ as its node and link sets\\
$P$ & The undirected and directed path set in $G_P$, where $\eta\in P$ and $p\in P$ denote the undirected and directed paths, respectively\\
$F$ & The NF set, where $f\in F$ denotes a network function\\
$V^{F}_P$ &  A set of all NF-enabled nodes, $V^{F}_P \subseteq V_P$\\
$\Gamma(F)$ & NF instance deployment, denoted as a tuple $[\{f,i, n^{f}_{i}\}: f\in F, i\in V_P]$, where $n^{f}_{i}$ is the instances of $f$ deployed onto $i$ \\
\hline
{Parameter} & {Description}\\
\hline
$\mathcal{D}, d_{st}$             & $\mathcal{D}$ is a set of service requests, where each  $d_{st}\in \mathcal{D}$  is a tuple $[(s,t), \sigma_{st}, F_{st}]$ representing one service request; here $F_{st}$ is the set of required NFs for demand $(s,t)$, and $\sigma_{st}=1$ if demand $(s,t)$ is SFC (i.e., $f\in F_{st}$ should be executed in a fixed sequence), otherwise, $\sigma_{st}=0$.\\ 
$\rho_i$ & The failure probability of NF-enabled physical node $i, i\in V^{F}_P$\\
\hline  
\end{tabular}
\end{table}

\subsection{Robust NF Service Evaluation Metric}
Our robust NF-service evaluation metric is based on the following observations.

\begin{observation}\label{ob:rqFulfill}
Given an NF-enabled node pool $V^{F}_{P}$ and requests $\mathcal{D}=\{d_{st}\}$, where request $d_{st}$ is realized through a path $\eta_{st}$. $d_{st}$ \emph{cannot} be fulfilled if and only if $V^{f}_P \cap \eta_{st}=\emptyset$, $f\in F_{st}$. 
\end{observation}
Observation~\ref{ob:rqFulfill} is derived from the fact that $d_{st}$ can only be fulfilled if and only if (all) the required NFs are deployed onto physical node(s) in its selected path $\eta_{st}$.

If $F_{st} = \{f\}$ and $\eta_{st}$ is given for all $(s,t)$, $Prob(d_{st})$, the probability of $d_{st}$ being fulfilled, is then ($1- \prod_{i\in V^{f}_{P}\cap \eta_{st}}\rho_{i}$).
We now consider a more generalized setting where demands are with single or multiple NFs and their routings $\eta_{st}$ are not given.
\begin{definition}\label{def:nonChainPrb}
Given NF-enabled node pool $V^{F}_{P}$, the \emph{robust NF-service evaluation metric}, denoted as $\mathcal{RP}(d_{st})$, is 
$$\mathcal{RP}(d_{st})= \min_{f\in F_{st}}\max_{\eta_{st}\in P_{st}}\left[1- \prod_{i\in \Gamma(f)\cap \eta_{st}}\rho_{i}\right].$$
\end{definition}
Note here that $d_{st}$ with multiple non-chained NF requests is fulfilled if and only if all required NFs are satisfied. Thus, the robust evaluation metric $\mathcal{RP}(d_{st})$ is determined by the worst best-case scenario among all requested NFs realized through the best-known paths in $P$. Hence, $\mathcal{RP}$ provides an \emph{estimated low bound} of NF-service reliability for all demands.

Different from non-chained NF requests, SFC request is fulfilled only when all required NFs are served in a specified sequence.
Without loss of generality, we assume that (1) \emph{the same NF request will not be fulfilled more than once on different NF-enabled nodes}, and (2) \emph{each NF-enabled node will not carry out multiple NF requests in SFC}.
\begin{definition}~\label{def:sfcPrb}
Given NF-enabled node pool $V^{F}_{P}$, the robust NF evaluation metric of SFC request $d_{st}$ is $$\mathcal{RP}(d_{st})=\min_{f\in F_{st}}\max_{p_{st}\in P_{st}}\left[1- \prod_{i\in \Gamma(f)\cap p_{st}}\rho_{i}\right] / |F_{st}|!.$$
\end{definition}

Since demands with SFC request are fulfilled only when all requested NFs are deployed onto $p_{st}$ and visited in a predefined sequence, there is only one valid case out of $|F_{st}|!$ permutations. $\mathcal{RP}(d_{st})$ is then determined by the worst best-case scenario among all requested NFs realized through the best-known paths in $P$ (with the highest probability to survive).

Considering multiple NF requests in a given NFVI (NFV infrastructure) and managed by the same NFV MANO, we define the robust NF-service evaluation metric among all NF request as follows.
\begin{definition}\label{def:nfvSrvPrb}
Given $G_P, G_L$, a set of NFs $F$, NF-enabled node pool $V^{F}_{P}$ and node failure probability $\rho_{i}, i\in V^{F}_{P}$, $\mathcal{RP}(V^{F}_{P})=\min_{d_{st}\in D} \mathcal{RP}(d_{st})$.
\end{definition}

Naturally, as the counterpart of robust NF evaluation metric, we may derive the following properties.
\begin{align}
&\mathcal{FP}(d_{st})= \max_{f\in F_{st}}\min_{\eta_{st}\in P_{st}}\left[\prod_{i\in \Gamma(f)\cap \eta_{st}}\rho_{i}\right]\label{fm:flENFReq}\\
&\mathcal{RP}(d_{st})=\max_{f\in F_{st}}\min_{\eta_{st}\in P_{st}}\left[\prod_{i\in \Gamma(f)\cap p_{st}}\rho_{i}\right] / |F_{st}|!\label{fm:flEsfcReq}\\
&\mathcal{FP}(V^{F}_{P})=\max_{d_{st}\in D} \mathcal{FP}(d_{st}) \label{fm:flEMetric}
\end{align}

\subsection{Illustrations: NF Service Reliability vs. Robust NF Service Evaluation Metric}
We evaluate the robust NF-service evaluation metric via an instance illustrated in Fig.~\ref{fig:nfRlb} and present its difference to the NF-service reliability defined in~\cite{ETSI2016Reliable}.
\begin{figure}[t!]
\centering
\includegraphics[scale=0.5]{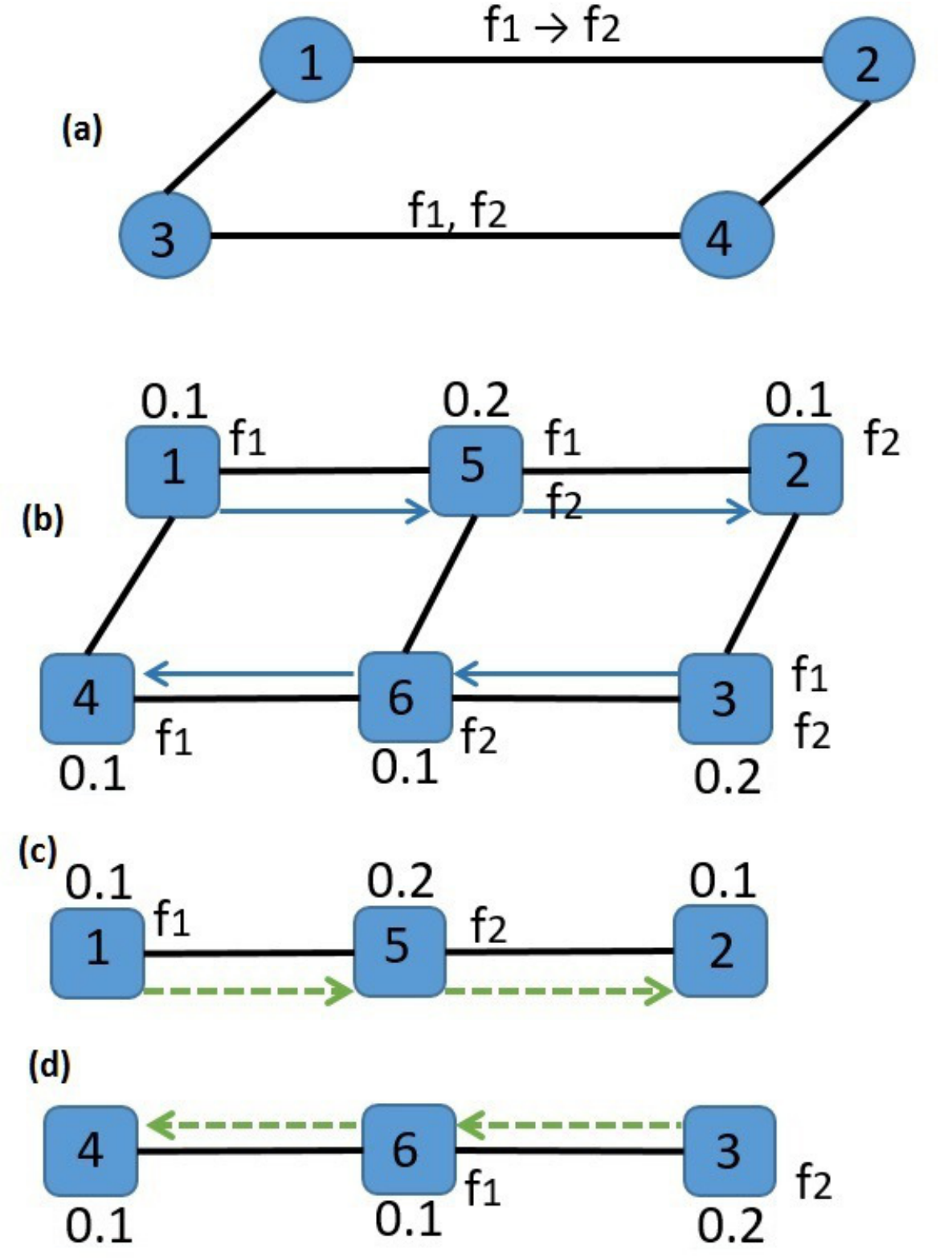}
\caption{NF reliability}
\label{fig:nfRlb}
\end{figure}
In this example, two demands with NF requests $d_{12}$ and $d_{34}$ are considered. Demand $d_{12}$ requires SFC $f_1 \rightarrow f_2$ and $d_{34}$ requires non-chained NFs $\{f_1,f_2\}$. NF-enabled nodes, their supported NFs, and their failure probabilities are labeled in Fig.~\ref{fig:nfRlb}. Candidate physical nodes to enable/deploy $f_1$'s  are in set $V^{1}_{P}=\{1,3,4,5\}$, and those for $f_2$'s are in $V^{2}_{P}=\{2,3,5,6\}$. $d_{12}$ is routed through  a directed path $\{(1,5), (5,2)\}$, and $d_{34}$ is routed through an undirected path $\{(4,6),(6,3)\}$. Based on the assumptions given in the previous section, the robust NF-service evaluation metric $\mathcal{RP}(\{d_{12},d_{34}\}) = \min\{1-0.1, 1-0.2, (1-0.2\times 0.1)/2, (1-0.1\times 0.2)/2\}= 0.49$.

Different from $\mathcal{RP}(\{d_{12},d_{34}\})$, NF-service reliability of $d_{12}$ is $[ 1- Prob(f_1, f_2 \text{ both failed})$ - $Prob(\text{only } f_2 \text{ failed})$ - $Prob(\text{only } f_1\text{ failed})$ - $Prob(f_1, f_2 \text{ fulfilled but not in-order}) ]$ = $1- 0.1*0.2*0.1 - 0.9*0.2*0.1- 0.9*0.2*0.1 - 0=0.962$.
The NF-service reliability of $d_{34}=[ 1- Prob(f_1, f_2 \text{ both failed})$ - $Prob(\text{only } f_2 \text{ failed})$ - $Prob(\text{only } f_1\text{ failed}) ]$ = $1 - 0.2*0.1*0.1- 0.2*0.1*0.9 - 0.2*0.9*0.1 = 0.962$.

The examples above show that NF-service reliability is measured when the deployment of NF instances and  routings are determined. In contrast, since the robust NF-service evaluation metric already evaluates the minimum NF reliability, the routings selected and the deployment of non-chained NFs or SFC would always be better than or at least equal to the metric. In other words, the robust NF-service evaluation metric provides a tight lower bound for each NF's reliability.

This instance also shows that with the limitation imposed on the NF-enabled nodes, the selection of NF-enabled nodes also impacts the robust NF-service evaluation metric. Hence, in the following section, we study the robust NF provisioning problem which aims at maximizing our proposed NF-service evaluation metric via NF-enabled node selection for NF request realization.

\subsection{Robust NF Routing}\label{subsec:routing}
We study the subproblems in robust NF provisioning, the NF $s-t$ path, SFC $s-t$ path, and request routing problems with the assumption that all NF-enabled nodes are deployed with corresponding NF instances.

\subsubsection{NF Path and SFC Path Problems}\label{subsubsec:sfcPath}
Different from typical $s-t$ routing problems in telecommunication networks, an SFC $s-t$ routing problem is to find a route realizing SFC request and guaranteeing that all required NFs are visited in order. The corresponding NF $s-t$ routing problem is to find a $s-t$ route visiting all required NFs enabled nodes. We formally define them as follows.
\begin{define}\label{def:nfPath}
Given demand $d$, required NFs $F(d)$, physical substrate network $G_P(V_P,E_P)$, and NF enabled node set $N_P(f)$, a \textbf{NF $s-t$ path} problem is to find a $s-t$ path $p_{st}$ where $N_P(f)\cup p_{st}\neq \emptyset$ with $f\in F(d)$.
\end{define}
\begin{define}\label{def:sfcPath}
Given demand $d$, its required SFC $(f_1, f_2, \cdots, f_r)$, physical substrate network $G_P(V_P,E_P)$, and NF enabled node set $N_P(f)$, a \textbf{SFC $s-t$ path} problem is to find a $s-t$ path $p_{st}$ which satisfies (1) $N_P(f)\cup p_{st}\neq \emptyset$ with $f\in F(d)$ and visiting of NF-enabled node in the same sequence of SFC.
\end{define}

Applying the robust NF evaluation metric, correspondingly, we introduce the robust counterparts of the above two problems as follows. With a single source-destination pair, the robust NF and SFC $s-t$ paths aim to find a $s-t$ path maximizing the minimal successful rate among all required NFs, i.e., $\max_{p\in P^{S}_{st}}  \min_{i\in p}\ln [1+(1-\rho_{i})]$ and $\max_{p\in P^{F}_{st}}  \min_{i\in p}\ln [1+(1-\rho_{i})]$, with $P^{S}_{st}$ and $P^{F}_{st}$ be the path sets for SFC $s-t$ path and NF $s-t$ path, respectively.

Given an SFC request, we consider its embedded SFC chain as a directed path in the logical network. While multiple demands with SFC requests are given, the required SFC chains together form an SFC forwarding graph (logical network). Let $G_L(V_L, E_L)$ indicate the logical SFC forwarding graph.
In Section~\ref{sec:formulation}, we present a labeling-based pseudo-polynomial algorithm for the SFC $s-t$ path and NF $s-t$ path on the constructed auxiliary NF-enabled network which is an intermediate network layer in between the lower-layer physical substrate network (namely, the physical network) and the upper-layer SFC chain (namely, the logical network).
Extending from a single source-destination pair of NF request, we consider the NF and SFC request routing problems in the next section.

\subsubsection{NF and SFC Request Routing}\label{subsubsec:nfRouting}
We still assume that all NF-enabled nodes are deployed with NF instances and study NF and SFC routing for all NF requests, whose corresponding general network design problem is the multi-commodity flow problem.

Given NF requests $D=\{d\}$, physical substrate network $G_P(V_P,E_P)$, and NF enabled node set $N_P(f)$, the \textbf{NF request routing} and \textbf{SFC request routing} problems are to generate NF/SFC paths for all demands with NF/SFC requests, respectively. Through adopting the robust NF evaluation metric as the objective, the robust NF and SFC routing problems determine the routes for all NF requests while evaluating the NF failure rate among all NFs and requests.

\subsection{Robust NF Provisioning}
With the evaluation metric above, we present in the following the \textit{VNF provisioning problem} without considering the NF-enabled node failures.
Given $G_P$, $G_L$, $\mathcal{D}$, and $V^{F}_{P}$. $d_{st}\in \mathcal{D}$, $s,t\in V_L$, is mapped onto a directed path $p_{st}\in P$ for SFC request, or undirected path $\eta_{st}\in P$ for NF request. We would like to determine a limited number of NF-enabled nodes to support each required NF and guarantee that demands are routed through their required NFs.
This problem considers both non-chained NF and SFC requests.

When taking the failures of NF-enabled nodes into consideration, we now define the robust VNF provisioning problem.
\begin{definition}
Given $N_f$ as the limited number of NF-enabled nodes supporting NF $f$,
the \textbf{robust VNF provisioning problem} is to determine the NF deployment which maximizes the robust NF-service evaluation metric:
$\max_{V^{f}_{P}: |V^{f}_{P}\leq N_f|}\mathcal{RP}(V^{F}_P)$.
\end{definition}

\subsection{Maximizing $\mathcal{RP}(V^{F}_{P})$ via Minimizing $\mathcal{FP}(V^{F}_{P})$}\label{subsec:maxminNF}
We show that the robust VNF provisioning \zblue{can be achieved via} finding the \emph{minimum robust NF failure evaluation metric}.

\begin{proposition}\label{prop:reqFailProb}
$1-\mathcal{RP}(V^{f}_{P})=\mathcal{FP}(V^{f}_{P})$, with $d_{st}\in \mathcal{D}$ and $f\in F_{st}$.
\end{proposition}
Derived directly from Definitions~\ref{def:nonChainPrb} and~\ref{def:sfcPrb}, Proposition~\ref{prop:reqFailProb} also holds for SFC requests.
Hence, we have the following conclusion.
\begin{theorem}\label{thm:evaluation}
$\max_{V^{F}_{P}}\mathcal{RP}(V^{F}_{P}) = \min_{V^{F}_{P}}\mathcal{FP}(V^{F}_{P})$.
\end{theorem}
In the next section, we demonstrate that solving the robust VNF provisioning via minimizing NF failure evaluation metric would linearize the non-linear equations. We then propose the solution approach accordingly.

\section{Solution Approach}\label{sec:formulation}
In this section, we present solution approaches to solve the robust VNF provisioning problem. We start with two special cases/subproblems: (1) the robust NF and SFC $s-t$ path problems, for which we construct an auxiliary network layer and present pseudo-polynomial algorithms for both; and (2) we leverage the $k$-level facility location problem and construct a 3.27-approximation algorithm for the VNF provisioning problem with SFC-Fork as the forwarding graph. We develop the two-step path-reduction techniques and demonstrate the existence of the approximation algorithm for the robust VNF provisioning problem through SFC-Fork. For the general problem setting (not limited to the SFC-Fork structure), we demonstrate how to utilize $\mathcal{FP}(V^{F}_{P})$ to formulate the robust VNF provisioning problem and propose its MILP solution approach. The variables and parameters used in this section are presented in Table~\ref{tbl:var}.

\begin{table}[ht]
\caption{Parameters and variables}\label{tbl:var}
\begin{tabular}{p{1.5cm}p{6.5cm}}
{Parameter}  & {Description}\\
\hline
$N_f$        & The number limitation of NF deployed locations with $f\in F$\\
$\rho_{i}$ & The failure probability of physical node $i$ with $i\in V_P$\\
$\delta^{i}_{\eta_{st}}$& A binary indicator showing whether physical node $i$ is on path $\eta_{st}$ or not, $\eta_{st}\in \mathcal{P}_{st}, (s,t)\in E_L$; if yes, $\delta^{i}_{\eta_{st}}=1$, otherwise $\delta^{i}_{\eta_{st}}=0$\\
$\gamma^{f}_{st}$ & A binary indicator showing whether $f$ is requested by $d_{st}$ or not; if yes, $\gamma^{f}_{st}=1$, otherwise, $\gamma^{f}_{st}=0$\\
$M$ & A very large number\\
\hline
{Variable} & {Description}\\
\hline
$\lambda$  & The upper bound of NF failure probability for service requests in $\mathcal{D}$\\
$\xi^{f}_{st}$ & NF failure probability of NF $f\in F$ and $d_{st}\in \mathcal{D}$\\
$x_{p_{st}}$    & A binary variable indicating whether path $p_{st}\in \mathcal{P}_{st}$ is selected to fulfill $d_{st}\in \mathcal{D}$\\
$y^{if}_{st}$    & A binary variable indicating whether physical node $i$ provides NF requests $f$ for $d_{st}$ or not; if yes, $y^{if}_{st}=1$, otherwise, $y^{if}_{st}=0$\\
$h_i$      & A binary variable which indicates whether a network function is deployed onto physical node $i$ or not; if yes, $h_i=1$, otherwise, $h_i=0$\\
$z^{f}_{i}$    & A binary variable which indicates if network function $f$ is deployed onto physical node $i$; if yes, $z^{f}_{i} = 1$, otherwise, $z^{f}_{i} = 0$\\
$\beta_{st}$    & A binary auxiliary variable which indicates if demand $d_{st}$ is selected under the SFC setting; if yes, $\beta_{st} = 1$, otherwise $\beta_{st} = 0$\\
\hline
\end{tabular}
\end{table}

\subsection{Special Case 1: NF and SFC $s-t$ Path Problems}\label{subsec:pseudo}
In this section, we present a pseudo-polynomial algorithm for the minimal weighted and robust non-chained NF $s-t$ path and SFC $s-t$ path problems, respectively.~\cite{vardhan2009finding} proved that a path with multiple must-stop nodes, without order requirements, is NP-Complete. The NF $s-t$ path is a path between $s$ and $t$ and must-stop at NF enabled nodes; and SFC $s-t$ path problem further required the NF path to visit NFs with defined order in SFC. Hence, we explore pseudo-polynomial algorithm for NF and SFC $s-t$ path problem. 

\subsubsection{Auxiliary NF-enabled Network Construction}\label{subsubsec:axNet}
We first introduce a condensed physical network, an \textbf{auxiliary NF-enabled network}, which only contains source and destination nodes of NF requests and their corresponding NF-enabled nodes. To allow a single node in the physical substrate network to support multiple types of NFs, we introduce augmentation steps that create copies of NF-enabled nodes and indicate their supported types of NFs in Algorithm~\ref{alg:nfAugNodeSet}.

\begin{algorithm}
\caption {Node set construction in the auxiliary NF-enabled network}
\label{alg:nfAugNodeSet}
\begin{algorithmic}[1]
	\REQUIRE Physical substrate network $G_P(V_P,E_P)$, NF-enabled node set $N_{P}(f)$ with $f\in F$ , and the initial augmented NF-enabled node set  $V^{A}_{P}=\emptyset$
	\ENSURE An augmented NF-enabled node set
	\FOR{$i\in N_P$ and $f\in F$}
		\IF{$i\in N_{P}(f)$}
			\STATE Create a copy of $i$ indicated as $i_{f}$
			\STATE $V^{A}_P=V^{A}_{P}\cup i_f$
		\ENDIF
	\ENDFOR
\end{algorithmic}
\end{algorithm}

We next present an algorithm that adds arcs in the auxiliary NF-enabled network through the cross-layer network concept, where we consider the SFC chain or SFC forwarding graph as the upper-layer/logical network and the physical substrate network as the lower-layer/physical network. After introducing duplicated NF-enabled nodes and their available NFs support, we build connections among these NF-enabled nodes based on the service requests. To limit the size of the augmented network,  we only add arcs connecting nodes in $V^{A}_{P}$ when the connection can realize the NF or SFC routes. We wish to note that the connectivity of the auxiliary NF-enabled network for non-chained NF requests is higher than that of the SFC version as the non-chained NF requests do not require in-order execution.the SFC version as the non-chained NF requests do not require in-order execution.

\begin{algorithm}
\caption{Arc construction with SFC requests in the auxiliary NF-enabled network}\label{alg:sfcArc}
\begin{algorithmic}[1]
	\REQUIRE Physical substrate network $G_P(V_P,E_P)$,  NF-enabled node set $N_{P}(f)$ with $f\in F$ , SFC forwarding graph $G_L(V_L,E_L)$, augmented NF-enabled node set $V^{A}_{P}$  and the initial augmented NF-enabled arc set  $E^{A}_{P}=\emptyset$
	\ENSURE Arc set for auxiliary NF-enabled network $E^{A}_{P}$ for SFC requests
	\FOR{$e=(f_i,f_j)\in E_L$}
		\FOR{$\ell \in N_P(f_i)$ and $k\in N_P(f_j)$}
			\IF{a path $\rho(\ell, k)$ exists in $G_P$}
				\STATE Create arc $(\ell(f_i), k(f_j))$ and add the arc into $E^{A}_{P}$
			\ENDIF
		\ENDFOR
	\ENDFOR
\end{algorithmic}
\end{algorithm}

\begin{algorithm}
\caption{Arc construction with non-chained NF requests in the auxiliary NF-enabled network}\label{alg:nfArc}
\begin{algorithmic}[1]
	\REQUIRE Given physical substrate network $G_P(V_P,E_P)$,  NF-enabled node set $N_{P}(f)$ with $f\in F$, demand $d_{st}$ with required NFs $D(F)$, augmented NF-enabled node set $V^{A}_{P}$,  and the initial augmented NF-enabled arc set $E^{A}_{P}=\emptyset$
	\ENSURE Arc set for auxiliary NF-enabled network $E^{A}_{P}$ for non-chained NF requests
	\FOR{any two $f_i,f_j\in D(F)$}
		\FOR{$\ell \in N_P(f_i)$ and $k\in N_P(f_j)$}
			\IF{A path $\rho(\ell, k)$ exists in $G_P$}
				\STATE Create arc $(\ell(f_i), k(f_j))$ and add the arc into $E^{A}_{P}$
			\ENDIF
		\ENDFOR
	\ENDFOR
 	\STATE $G_S = G_S \cup \{s, t\}$
	\STATE $\rho_s=0$, $\rho_t=0$
 	\FOR{$i\in N_P(f_1)$}
		 \IF{a path $\rho(s,i)$ exists in $G_P$}
			 \STATE Create arcs $(s, i)$ and add into $E_S$
		 \ENDIF
	 \ENDFOR
 	\FOR{$i\in N_P(f_r)$}
		 \IF{a path $\rho(i,t)$ exists in $G_P$}
			 \STATE Create arcs $(i,t)$ and add into $E_S$
		 \ENDIF
	 \ENDFOR
\end{algorithmic}
\end{algorithm}

To differentiate the auxiliary NF-enabled network for NF requests and SFC requests, we let $G_S(N_S, E_S)$ and $G_F(N_F, E_F)$ denote networks for non-chained NF requests and SFC requests, respectively, where $N_S=V^{A}_{P}$, $N_F=V^{A}_{P}$, and $E_S$ and $E_F$ are obtained through Algorithms~\ref{alg:sfcArc} and~\ref{alg:nfArc}, respectively.
We illustrate an instance of the auxiliary NF-enabled network for SFC in Fig.~\ref{fig:cnctPhy},     which is an abstraction of all potential SFC path realization via NF-enabled physical nodes.
Given SFC $(f_1, f_2, f_3)$ (see Fig.~\ref{fig:cnctPhy}(a)), where $f_1$, $f_2$, and $f_3$ are with 2, 3, and 2 NF-enabled physical nodes, respectively, which are illustrated in Fig.~\ref{fig:cnctPhy}(b). All possible SFC physical paths for $(f_1, f_2, f_3)$ through their corresponding NF-enabled nodes can be calibrated. For instance, there are 12 possible physical paths to realize the SFC in Fig.~\ref{fig:cnctPhy}.

\begin{figure}
\centering
\includegraphics[width=8.7cm]{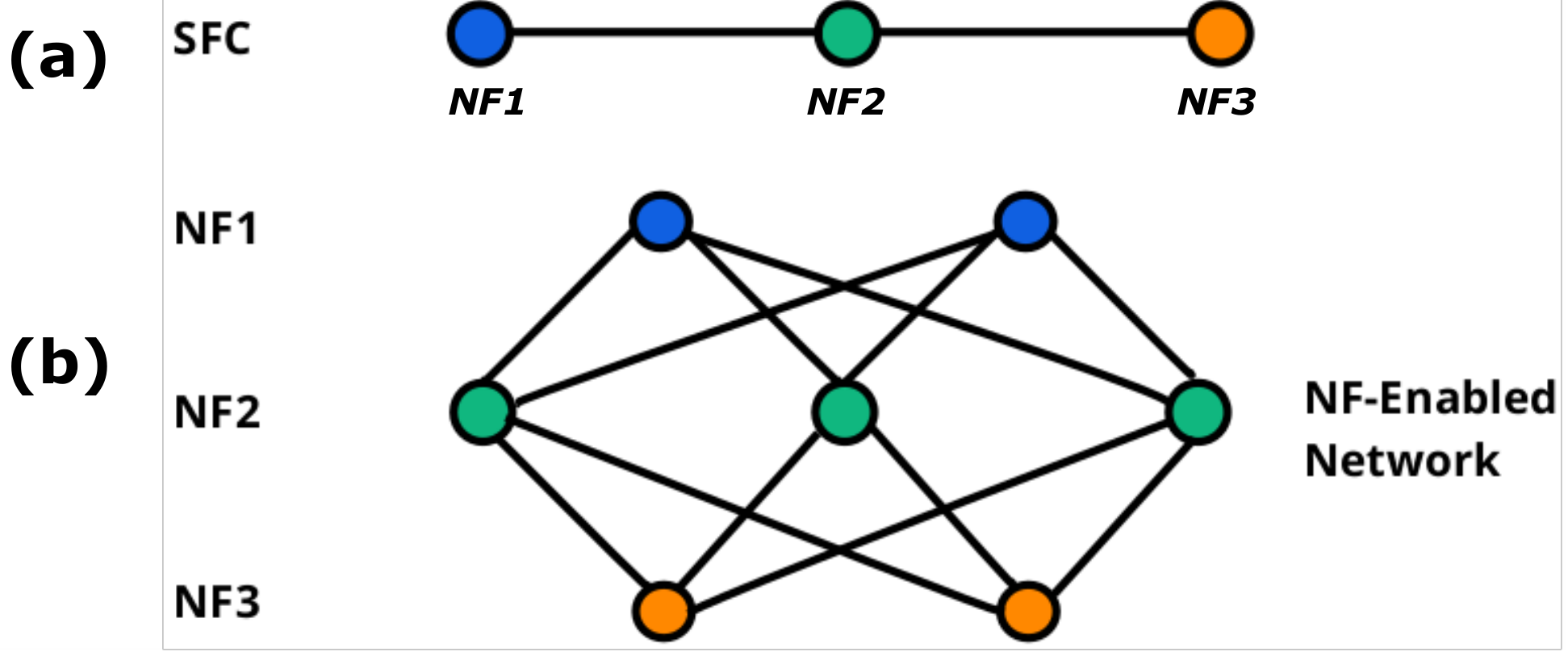}
\caption{Auxiliary NF-enabled network for SFC}
\label{fig:cnctPhy}
\end{figure}

\begin{proposition}\label{prop:sfcorder}
Given a constructed $G_S(N_S,E_S)$, the visiting sequence of NFs in SFCs is realized through the arcs in $E_S$.
\end{proposition}

\nproof
We prove this claim by contradiction. Given an SFC $\lambda$ and its corresponding $G_S(N_S,E_S)$ and let $(f_i, f_j)$ and $(f_j, f_k)$ be arcs in $\lambda$. We assume that $E_S$, arc $(i_{f_i}, i_{f_k})$ exists. Based on Algorithm~\ref{alg:sfcArc}, only $(i_{f_i}, i_{f_j})$ and $(i_{f_j}, i_{f_k})$ are created. Hence, connecting $i_{f_i}$, $i_{f_k})$ requires at least two arcs in $G_S$. Contradiction!
\proofend

\subsubsection{Pseudo-Polynomial Algorithm for SFC $s-t$ Path}\label{subsubsec:algoSFCpath}
We first define the SFC $s-t$ path problem before presenting the algorithm.
Given physical substrate network $G_P(V_P, E_P)$ and all of its supported NFs, we let $N_{P}(f)$ represent an NF-enabled physical node set supporting network function $f\in F$.
\begin{define}\label{def:sfcPhyPath}
Given an SFC chain, \textbf{SFC service path} is a physical path connecting NF-enabled nodes with required NFs following the sequence defined in the SFC chain.
\end{define}
We let $\mathcal{P}(f_h)$ represent a subpath set of SFC physical paths in $G_P(V_P, E_P)$, which starts from physical nodes in $N(f_h)$ and ends at physical nodes $N(f_r)$.

We now present the algorithm for SFC $s-t$ path. Given $G_P, G_S, G_L$ and the source and destination nodes $s$ and $t$ of the SFC. For the general minimal-weighted SFC $s-t$ path problem, we first add $s$ and $t$ connecting the first and last NF-enabled nodes in $G_S$. Here, the weight can be the shortest-path weight in the physical substrate network. Different from the minimal-weighted $s-t$ path, the SFC $s-t$ path should visit the required NFs in the order specified in SFC. Since Proposition~\ref{prop:sfcorder} shows that the order of NFs in SFC is preserved in the auxiliary NF-enabled network, we present a Dijkstra-like algorithm for SFC $s-t$ path as follows.

\begin{algorithm}
\caption{SFC $s-t$ path algorithm}\label{alg:sfcPath}
\begin{algorithmic}[1]
	\REQUIRE Given $G_P(V_P,E_P)$, $G_S(N_S, E_S)$, and SFC chain $G_L(V_L,E_L)$,  source and destination node $s, t$
	\ENSURE SFC path $\rho$
	\FOR{$i\in N_S$}
		\STATE Set initial visited ancestor list $\ell(i)=\emptyset$
	\ENDFOR
	\STATE Set a node set $U=\{s\}$, $\omega(s)=0$
	\WHILE{$U \neq V_P$}
		\FOR{$j$ in $U$'s adjacent nodes}
			\STATE $\omega(j)= \min_{e=(v,j): v\in U}[\omega(v)+\omega(e)]$
		\ENDFOR
		\STATE{Add $j*=\arg\min_{j\in U\text{'s adjacent nodes}}\omega(j)$ to $U$}
	\ENDWHILE
\end{algorithmic}
\end{algorithm}

\subsubsection{Pseudo-Polynomial Algorithm for Robust SFC $s-t$ Path}\label{subsubsec:algoRbSFC}
We present in this section the pseudo-polynomial algorithm for the robust SFC $s-t$ path problem.
We identify the property of the robust SFC $s-t$ path and its non-chained counterpart as follows.

\begin{proposition}\label{prp:btlPath}
Given an auxiliary NF-enabled network $G_S(N_S, E_S)$ and $G_F(N_F, E_F)$\\
(1) The robust SFC $s-t$ path problem is an \textbf{SFC $s-t$ bottleneck path problem} in $G_S(N_S, E_S)$. \\
(2) The robust non-chained NF $s-t$ path problem is a \textbf{non-chained NF $s-t$ bottleneck path problem} in $G_F(N_F, E_F)$.
\end{proposition}

\nproof
The $s-t$ bottleneck path problem determines a path with a maximal path capacity defined as the minimal edge capacity on the path.
Let all arcs in $G_P$ have failure probability equals 0. Then, we do the typical arc augmentation for all NF-enabled nodes, where augmented nodes are added and directed arc are created to connect these nodes. For the SFC path, the arc direction follows the SFC chain; as for the non-chained NF path, the arcs are bi-directly generated. All augmented arcs have capacity $ln[1+(1-\rho_i)]$. Hence, the robust SFC $s-t$ and non-chained NF $s-t$ paths become the corresponding bottleneck path problems.
\proofend

\begin{algorithm}
\caption{Robust SFC  $s-t$ path algorithm}\label{alg:robust-sfcPath}
\begin{algorithmic}[1]
	\REQUIRE Given $G_P(V_P,E_P)$, $G_S(N_S, E_S)$, and SFC chain $G_L(V_L,E_L)$,  source and destination node $s, t$
	\ENSURE SFC path $\rho$
	\FOR{$e\in E_S$}
		\STATE Set $\rho_{e}=-1$
	\ENDFOR
	\FOR{All NF-enabled nodes in $N_S$}
		\STATE Augment all NF-enabled nodes as arcs and added into $E_S$
		\STATE Set augmented arc capacity $\omega(e)= ln (1+ (1-\rho_i))$
	\ENDFOR
	\STATE Set a node set $U=\{s\}$ and $\omega(s)=0$
	\WHILE{$U \neq V_P$}
		\FOR{$j$ in $U$'s adjacent nodes}
			\STATE $\omega(j)= \max_{e=(v\in U,j}[\min (\omega(v), \omega(e))]$
		\ENDFOR
		 \STATE $U=U\cup \{i\}$, while $i=argmin_{j\in U\text{'s adjacent nodes}}\omega(j)$
	\ENDWHILE
 \end{algorithmic}
 \end{algorithm}


For non-chained NF requests, we assume that only nodes in the auxiliary NF-enabled network have nodal weights and all arcs have weight equals 0. We show that the optimal solution of an SFC $s-t$ path is also the optimal solution of its non-chained counterpart using the objective $\min \sum_{i\in \cup_{f\in F}i_{f}}c_{i}x_{i}$ which minimizes the node-weighted SFC $s-t$ path. We demonstrate that if NF-enabled nodes are reached to achieve the optimal solution, the visiting order of these nodes would not impact that optimal solution.visiting order of these nodes would not impact that optimal solution.

\begin{proposition}
Given $G_S(N_S,E_S)$ and two $s-t$ NF paths $p_1$ and $p_2$ containing NF-enabled node $N_1$ and $N_2$, where $p_1\neq p_2$ if $N_1=N_2$. We have $\min \sum_{i\in \cup_{f\in F}i_{f}\cap p_1 }c_{i}x_{i}=\min \sum_{i\in \cup_{f\in F}i_{f}\cap p_2 }c_{i}x_{i}$.
\end{proposition}

Hence, we define an SFC for the non-chained NF $s-t$ path and apply Algorithms~\ref{alg:sfcPath} and~\ref{alg:robust-sfcPath} to obtain the (robust) non-chained NF $s-t$ path.

\subsection{Special Case 2: Robust NF Provisioning with SFC-Fork}\label{sec:appAlgNF}
In this section, we present an approximation algorithm for robust VNF provisioning with \textbf{SFC-Fork} (also denoted as \textbf{SFork}). Note here that SFC-Fork is a common NF forwarding graph defined in practice, which has a rooted tree structure with a single branching point. We first review the existing $k$-level facility location bi-factor approximation algorithm and demonstrate that VNF provisioning with a single SFC can be reduced to a $k$-level facility location problem. We then design an approximation algorithm for the problem. To manage the SFC-Fork as the NF forwarding graph, we apply a two-step parameterized path reduction in the bi-factor approximation algorithm and proof that it is a 3.27-approximation. Leveraging the approximation algorithm for robust facility location problems, we further demonstrate the existence of the approximation algorithm for the robust VNF provisioning problem.

\tred{We let $G_S(V_S, A_S)$ represent an SFC forwarding graph, $(f_1, f_2, \cdots, f_r)$ be an SFC chain, and $\Lambda=\{\lambda\}$ denote a set of SFC chains with $\lambda\in\Lambda$ as an SFC chain.}
\tred{Given physical substrate network $G_P(V_P, E_P)$ and all of its supported NFs, we let $N_{P}(f)$ represent an NF-enabled physical node set supporting network function $f\in F$.}
\tred{We let $\mathcal{P}(f_h)$ represent a subpath set of SFC physical paths in $G_P(V_P, E_P)$, which starts from physical nodes in $N_{P}(f_h)$ and ends at physical nodes $N_{P}(f_r)$.}

\begin{assumption}\label{aspt:2conn}
We assume that $G_P(V_P,E_P)$ is at least two-connected. 
\end{assumption}
With Assumption~\ref{aspt:2conn}, the networks created through Algorithms~\ref{alg:nfAugNodeSet} - \ref{alg:nfArc} are at least two-connected. 
\begin{proposition}\label{prop:2connNFNet}
With Assumption \ref{aspt:2conn}, if $|N_P(f)|\geq 2$ for $f\in F$, $G_S(V_S,E_S)$ is at least two-connected.
\end{proposition}

\subsubsection{Review on $k$-level Facility Location Problem}\label{subsec:klevel}
Given a client set $D$ and a facility set $F_{\ell}$ at level $\ell$, the $k$-level facility location problem determines the sets of facilities $X_{\ell} \in F_{\ell}$ to be opened at level $1\leq \ell \leq k$ and connects client $d\in D$ to a \textbf{facility service path} $(i_k(d), i_{k-1}(d),\cdots, i_{1}(d))$ with facility location at $i_{k}(d)$. The 1-level facility location problem only has a single level of facility set, where all clients are directly connected to the facility location without a service path.

\cite{ageev2004improved} demonstrates that an instance of $k$-level facility location problem can be reduced to an instance of 1-level facility location problem as follows: the 1-level problem takes the client set in the $k$-level problem as its client set, and the facility location set is determined by the potential facility service paths from level $k$ to level $1$, which is denoted as 
\begin{align*}
\rho(i_k, t) = \arg\min_{\rho\in \mathcal{P}_{F}}\{t\times \beta \times C(\rho)+ \alpha \times O(\rho)\},
\end{align*} 
where $\mathcal{P}_{F}$ is an NF service path set and $t=1, \cdots, |D|$, $i_{k} \in X_k$.
Client $j\in D$ can be connected to these determined service paths with connection cost $C(j, i_k)+ C(\rho(i_k,t))$. Given a solution of the 1-level facility location problem (denoted as \textbf{SOLS}) constructed above, a corresponding $k$-level facility location solution (denoted as \textbf{SOLM}) can be constructed through opening all facilities on above facility services paths and connecting all clients with their corresponding service paths.
\begin{theorem}[Theorem 1 in~\cite{ageev2004improved}]\label{thm:klevel}
If SOLS is an $(\alpha, \beta)$-approximate solution of an 1-level facility location instance, then, SOLM is a $(\alpha, 3\beta)$-approximate solution of a $k$-level facility location instance.
\end{theorem}

Next, we present approximation algorithms for NFP-SFork starting with a simple case, where a single SFC is the NF-forwarding graph. In other words, all SFC requests require the same SFC. Extending from this special case, we present approximation algorithms with SFork as the forwarding graph in both deterministic and robust settings. 

\subsubsection{Approximation Algorithm for NFP-1SFC}\label{subsubsec:simplecase}

We first study a special case of the NF provisioning problem, where all requests require the same SFC $\lambda=(f_1, f_2, \cdots, f_{\gamma})$ and $\gamma$ indicates the $\gamma$th NF in the SFC. We denote the problem as \textbf{NFP-1SFC}.

NFP-1SFC can be reduced to a $k$-level facility location problem through the following steps. First, we convert an SFC path set into the connection between a request and its SFC service path, where required NFs specified in the SFC should be visited in order along the path.
We calibrate the cost of request $d\in D$ to its SFC service path with a simple reduction. The connection cost for a request $d(s,t)$ to a path SFC service $\rho$ starting from $i_{f_{\gamma}}$ and ending at $i_{f_1}$ is
$$C((s,i_{f_1}), \rho, (i_{f_{\gamma}},t)) = C(s, i_{f_1}) + C(\rho)+C(i_{f_\gamma},t)$$
$$=C(d, i_{f_1},i_{f_\gamma}, \rho)+C(\rho),$$ with $C(d,i_{f_1},i_{f_\gamma}, \rho)=C(s, i_{f_1})+C(i_{f_\gamma},t)$, $i_{f_1}, i_{f_{\gamma}}\in \rho$ and $\rho \in \mathcal{P}_{F}$. The procedure is presented in Algorithm~\ref{alg:demandCon}. Based on $G_S(V_S, E_S)$'s 2-connectivity given in Proposition~\ref{prop:2connNFNet}, the connection between the source and destination nodes of a request and a service path's two-end nodes can be established.

\begin{algorithm}
\caption{SFC routing conversion to request and service path connections}\label{alg:demandCon}
\begin{algorithmic}
\STATE{Input:}
\FOR{$d=(s,t)\in D$, $i_{f_1}\in X_{1}$, and $i_{\gamma}\in X_{\gamma}$}
\FOR{$\rho\in \mathcal{P}_{F}$}
\STATE{Calibrate the shortest path between $(s,i_{f_1})$ and $(t,i_{f_r})$ in $G_P$ with $i_{f_1}, i_{f_r}\in \rho$}
\STATE{Set $C(d,i_{f_1}, i_{f_r}, \rho)=C(s, i_{f_1})+C(i_{f_r},t)$}
\ENDFOR
\ENDFOR
\end{algorithmic}
\end{algorithm}

After applying Algorithm~\ref{alg:demandCon}, the NFP-1SFC turns to a problem that (1) determines NF-enabled nodes at each level for NF instance deployment, and (2) guarantees all requests are connected to an NF service path visiting NFs in the order defined in the SFC. Hence, the problem is a $\gamma$-level facility location problem, where (1) the NF request set is the client set, and (2) the NF-enabled node set corresponds to NF $i$ are the level $i$ facility set ($1\leq i \leq \gamma$). The difference between the two is that the connection cost of a request to a service path is composed of two parts, namely $C(s, i_{f_1})+C(i_{f_r},t)$.
Different from the $k$-level facility location problem, we let $\overline{\mathcal{P}_{S}}$ be the service path set for request $d(s,t)$ which starts and ends at node $i_{f_r}$ and $i_{f_1}$, respectively, with $i_{f_r} \in X_{\gamma}$ and $i_{f_1}\in X_{1}$. We select service path $\rho(t, i_{f_1}, i_{f_\gamma}, d)$ for request $d\in D$ as
$$\arg\min_{\rho\in \overline{\mathcal{P}_{S}}}\{t\times \beta \times C(\rho)+ \alpha \times O(\rho)\},$$
where $t=1,\cdots, |D|$. We let request $d(s,t)$ connect to NF service path $\rho(t, i_{f_1}, i_{f_\gamma}, d)$.

We further reduce NFP-1SFC to the 1-level facility location problem, where the SFC request set is taken as the client set, and the selected SFC service path set represents the facility set.
Given a feasible solution of the 1-level facility location problem, denoted as $\psi_{1FL}$, we construct a solution for NFP-1SFC problem as follows -- all NF-enabled nodes on selected SFC service paths are deployed with NF instances, and the demand $d(s,t)$ is connected to a selected service path.

With Theorem~\ref{thm:klevel}, the following conclusions holds.
\begin{lemma}\label{lm:forest}
Given an SFC chain $\lambda=\{f_1,f_2,\cdots, f_{\gamma}\}$ and a feasible solution $\psi_{1FL}$, \\
(1) if there exist NF service paths $\rho_{1}(\psi_{1FL})=(i_{f_{1}}, i_{f_{2}},\cdots, i_{f_{\gamma}})$, $\rho_{2}(\psi_{1FL})=(i'_{f_{1}}, i'_{f_{2}},\cdots, i'_{f_{\gamma}})$ and $i_{f_j} = i'_{f_j}$, another solution $\phi$ also exists based on SOLS, where NF service paths are
$$\rho_{1}(\phi)=(i_{f_{1}}, i_{f_{2}},\cdots, i_{f_j},\cdots, i_{f_{\gamma}})$$ 
and 
$$\rho_{2}(\phi)=(i'_{f_{1}}, i'_{f_{2}},\cdots, i'_{f_j}, \cdots, i'_{f_{\gamma}}), \text { with } i_{f_{\ell}}=i'_{f_{\ell}}, \ell \leq j;$$ \\
(2) $C^{\text{SOLS}}_{1}=C^{\phi}_{1}$; $O^{\text{SOLS}}=O^{\phi}$; and $\sum^{\gamma}_{i=2}C^{\text{SOLS}}_{i} \leq \sum^{\gamma}_{i=2}C^{\phi}_{i}$.
\end{lemma}

We derive from the first claim that the solutions of the NFP-1SFC problem have SFC service paths satisfying all NF requests ($(f_1, f_2,\cdots, f_r)$), which form a forest beginning at the first NF-enabled node. Hence, a feasible solution for NFP-1SFC also has the forest structure. Meanwhile, all NF-enabled nodes and arcs in the forest, formed by its SFC service paths, are unique. The NF deployment cost for NF provisioning is the total cost to deploy all NF-enabled nodes in the forest, and the connection costs are the total costs of arcs in the forest. 
Moreover, with SFC path selection, the Claim 2 above further identifies that demands connected to a rooted tree of the forest have their source nodes connected to the root node. Hence, we have a similar conclusion on the bi-factor approximation for NF provisioning with a single SFC.

\begin{theorem}\label{thm:sfcChainApp}
If the 1-level facility location problem has a $(\alpha, \beta)$-approximation solution, a solution of NFP-1SFC can be constructed through a $(\alpha, 3\beta)$-approximation.
\end{theorem}
The detailed proof of Theorem~\ref{thm:sfcChainApp} 
is very similar to that of Theorem~\ref{thm:klevel} in \cite{ageev2004improved}. Since NFP-1SFC is a special case of NFP-SFork, we include the proof of Theorem 10 
in Appendix~\ref{sec:bifactSFCproof} for NFP-SFork.

\subsubsection{Approximation Algorithm for NFP-SFork}~\label{subsubsec:bifactAlg}
With multiple SFCs, the corresponding SFC forwarding graph identified in \cite{cheng2015enabling, jalalitabar2018service, nfv2014001, openSFC} forms a fork network structure. Extending from NFP-1SFC, we study the approximation algorithm for SFC-Fork, where multiple SFC chains share the common structure. We define the SFC-Fork as follows.
\begin{define}\label{def:fork}
Given $d\in D$ and all requested SFCs in $\Lambda$, an SFC-Fork is a rooted tree with a single branching node $f_{b}$, where $f_{b}\in \cup_{f\in \lambda}\cup_{\lambda\in \Lambda} f$. 
\end{define}
\begin{figure}
\centering
\includegraphics[width=8.7cm]{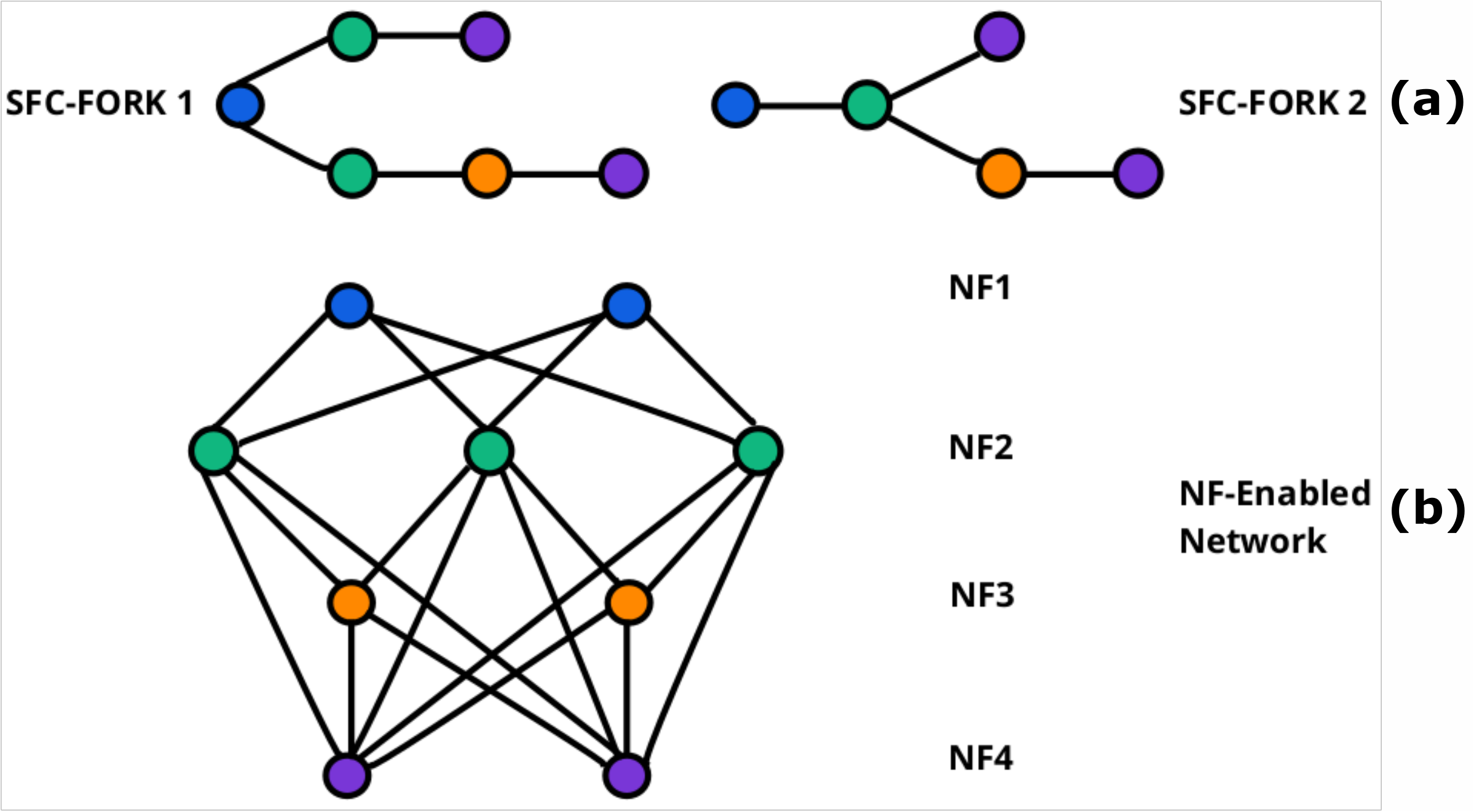}
\caption{Fork SFC forwarding graph and NF-enabled node connection graph}
\label{fig:sffork}
\end{figure}
Figure~\ref{fig:sffork} illustrates two instances of SFC-Forks, where (1) the (sub)paths in SFC-FORK1 share a common NF -- the root node of the fork; and (2) the SFC-FORK2 is branched at the NF2 (green) node. All SFCs visit the same NF with the same order. For example, NF 2 is visited before NF 4 for all SFCs. In other words, edges going from NF 4 to NF 2 are not allowed (only top-down order as in Fig.~\ref{fig:sffork}(b)). Hence, a forwarding graph of an SFC-Fork is $\cup_{\lambda\in \Lambda}\lambda$.
The difference between NFP-SFork and the $k$-level facility location problem is that the latter only contains a type of facility service paths, while NF-SFork requires multiple types of SFC paths with shared NFs. Thus, NF-SFork requires separate management to avoid duplicated counts either in its solution or its reduced 1-level facility location problem on NF deployed nodes. We now show that a bi-factor $(\alpha, 3\beta)$-approximation algorithm exists for NFP-SFork if the 1-level facility problem has a $(\alpha,\beta)$-approximation.
$(\alpha,3\beta)$-approximation algorithm.

Given an instance of NFP-SFork $\mathcal{M}$, we let the common subpath of SFCs be $\rho^{S}_{\Lambda}=(f_1,\cdots, f_{b})$. If $b=1$, only the first NF is shared; otherwise a subpath $(f_1,\cdots, f_{b})$, $2\leq b \leq \min_{\lambda\in \Lambda}\gamma(\lambda)$, is shared, where $\gamma(\lambda)$ indicates the total number of NFs in SFC $\lambda\in\Lambda$.
We reduce the NFP-SFork to the 1-level facility location problem through the following approaches -- (1) facility set: reducing the shared sub-SFC service paths, and (2) client set: aggregating SFC requests and their remaining subpaths in SFC paths.

Algorithm~\ref{alg:twoStepPathReduction} is a \textbf{two-step parameterized path reduction algorithm} for creating SFC service paths.
\begin{algorithm}
\caption{Two-step parameterized path reduction algorithm}
\label{alg:twoStepPathReduction}
\begin{algorithmic}
\STATE \textbf{Step 1:} \emph{Parameterized path reduction of SFC service subpaths from $f_{b+1}(\lambda)$ to $\gamma(\lambda)$ with $\lambda\in \Lambda$.} Let $\bar{\mathcal{P}(b+1,\gamma(\lambda))}$ be a path set on the auxiliary NF-enabled node network, $G_A$, connecting $f_{h+1}$-enabled nodes with $b+1 \leq h \leq \gamma(\lambda)$ and $$p(t,i_{f_{b+1}})=\arg\min_{p\in \bar{P(b+1,\gamma(\lambda))}} \{t \beta \left[O(p) +  C(p)\right]\}$$, $t=1,\cdots, |D|$. This path set is called the \textbf{disjoint SFC service subpath}, where $O(p)$ and $C(p)$ are NF deployment costs and connection costs for path $p$, respectively.
\STATE \textbf{Step 2:} \emph{Parameterized path reduction of the joint path among all SFCs}. We combine the joint path and subpaths in $\mathcal{P}(b+1,\gamma(\lambda))$ and construct full SFC paths for requests. We determine the shared path as $p(j, i_{f_1})=\arg\min_{p\in \bar{\mathcal{P}(1,b)}} \{\alpha O(p) + j \beta C(p)\}$ with $j=1, \cdots, |N(f_{b+1})|$, where $\bar{\mathcal{P}(1,b)}$ is the subpath set connecting  the first NF-enabled nodes all the way through the $f_{b}$-enabled nodes in $G_A$. Hence, given a request $d\in D$, a feasible solution of NFP-SFork has the following structure -- request $d$ connects to path $p(t,i(f_{b+1})$ and $p(j, i_{f_1})$, where $t$ is the request, $d$ is an index, and $j$ is the index of the $f_{b+1}$-enabled node (denoted as $i(f_{b+1})$).
\end{algorithmic}
\end{algorithm}
The shared subpath among requests is called the \textbf{shared SFC service subpath} for all $j\in N(f_{b})$.
The procedure of \textbf{NFP-SFork reduction to 1-level facility location} is now presented as follows.\\
1. Facility location set: it contains all $p(j, i_{f_1})$, $j=1,\cdots,\bar{\mathcal{P}(1,b)}$, and has setup costs equal the sum of deployment and connection costs of their corresponding paths.\\
2. Client set: NF $(b+1)$-enabled nodes, $i_{f_{b+1}}$,  where subpaths and SFC requests connecting to them are aggregated with $i_{f_{b+1}}\in X_{b+1}$. The connection cost to a facility location is the sum of (1) the connection costs from request $d$ to a disjoint subpath, (2) the connection costs from disjoint SFC subpath to shared subpath, and (3) the deployment and connection cost of disjoint SFC service subpath, which is captured through
$$C(d,t)+C(p(t,i(f_{b+1}))+C(p(i_{f_{b+1}},i_{f_1}))$$
$$+C(i_{f_{b+1}},i_{f_1})+F(p(t,i(f_{b+1})).$$

We let $O^{1}$ indicate the total NF deployment cost corresponding to 1-level facility location facility cost, and $O^{2}$ indicates the total NF deployment cost which is part of 1-level facility location connection costs.
Based on an instance of the 1-level facility location problem constructed above and one of its feasible solutions, we get an \textbf{NFP-SFork feasible solution} by \\
1. Following the facility and client connections -- connect (1) the shared SFC service subpaths and disjoint SFC service subpaths, and (2) SFC requests to its disjoint subpaths.\\
2. Deploying NFs onto NF-enabled nodes for both shared and disjoint SFC service subpaths.

We now discuss the existence of a $(\alpha, 3\beta)$-approximation algorithm for NFP-SFork. The proof of Theorem 10 
is given in the appendix.

\begin{theorem}\label{thm:forkAPP}
Given a $(\alpha, \beta)$-approximation solution for the 1-level facility location problem, there exists a $(\alpha, 3\beta)$-approximation algorithm for the NFP-SFork problem.
\end{theorem}
If the following inequalities hold individually, it sequentially leads to Theorem~\ref{thm:forkAPP}, where $\varphi_{\text{1fl}}$ is a feasible solution of the 1-level facility location problem, and $\varphi_{\text{sfc}}$ is a feasible solution constructed based on $\varphi_{\text{1fl}}$ for NFP-SFork.
\begin{align}
O(\varphi_{\text{sfc}})+C(\varphi_{\text{sfc}})&\leq F(\varphi_{\text{1fl}})+C(\varphi_{\text{1fl}}) \label{forkproof:1app}\\
                           &\leq  \alpha F(\psi_{\text{1fl}})+ \beta C(\psi_{\text{1fl}})\label{forkproof:1fl}\\
                          &\leq  \alpha O(\psi_{\text{sfc}})+ 3\beta C(\psi_{\text{sfc}})\label{forkproof:forkCost}
\end{align}
Based on the assumption that the 1-level facility location problem has a $(\alpha, \beta)$-approximation algorithm, inequality~(\ref{forkproof:1fl}) holds. The proofs of inequalities~(\ref{forkproof:1app}) and~(\ref{forkproof:forkCost}) are given in Lemma~\ref{lm:forkCost} and~\ref{lm:fork13app}, respectively in Appendix~\ref{sec:bifactSFCproof}, where a supporting conclusion presented in Lemma~\ref{lm:forkForest} of Appendix~\ref{sec:bifactSFCproof} shows that a forest structure exists in the joint SFC paths for requests of an SFC $\lambda\in \Lambda$.

Next, we present the bi-factor algorithm for NFP-SFork as follows.
\begin{algorithm}
\caption{Bi-factor approximation algorithm for NFP-SFork}\label{alg:greedyTwoStepPath}
\begin{algorithmic}
\STATE Greedy Algorithm~\cite{ageev2004improved}
\STATE Step 1: Given a single level facility location problem, we scale the facility open cost up with a ratio $\delta$ with $\delta \geq 1$
\STATE Step 1.1: Initially, set $B_j=0$ for all clients. Assign budget $B_j$ to all clients $j$, and client $j$ offers $\max\{B_{j}-c_{ij},0\}$ to facility $i$ if $j$ is not connected, otherwise, $\max_{i'\neq i}\{c_{i'j}-c_{ij},0\}$ if the client $j$ connects to a facility $i'$.
\STATE Step 1.2: If unconnected client set $I_{U}\neq \emptyset$, increase $B_j$ at the same rate; if the total offered costs to unopened facility is equal to open costs, i.e., $\sum_{j\in I_{U}}\max\{B_{j}-c_{ij},0\}=f_i$, open facility $i$; and if the connection costs of unconnected client $j$ equals its connection cost to an opened facility $i'$ $\max_{i'\neq i}\{c_{i'j}-c_{ij},0\}=c_{i'j}$, connect client $j$ to facility $i$.
\STATE Step 2: Scale down the open costs of facilities to their original costs at the same rate; if opening a facility does not increase the total cost, the facility is open and assign clients to its closest open facility.
\STATE Algorithm~\ref{alg:twoStepPathReduction}: two-step path reduction 
\end{algorithmic}
\end{algorithm}
A greedy algorithm presented in~\cite{ageev2004improved} is a bi-factor $\gamma_{f}(\delta), \gamma_{c}(\delta)$ approximation algorithm for the single-level facility location problem, where $\gamma_{f}(\delta)=\gamma_{f}+\ln(\delta)$ and $\gamma_{c}(\delta)=1+\frac{\gamma_{c}-1}{\delta}$ with $\gamma_{f}=1.11$ and $\gamma_{c}=1.78$. 
Combining the greedy algorithm in~\cite{ageev2004improved} and the proposed Algorithm~\ref{alg:twoStepPathReduction}, we obtain a bi-factor approximation algorithm (Algorithm~\ref{alg:greedyTwoStepPath}) for NFP-SFork. 
\begin{theorem}~\label{thm:327appAlg}
Algorithm~\ref{alg:greedyTwoStepPath} is a 3.27 approximation algorithm for NFP-SFork problem regardless of forward graph's network structure.
\end{theorem}
Based on Theorem~\ref{thm:forkAPP} and the $(\gamma_{f}(\delta), \gamma_{c}(\delta))$-approximation algorithm for 1FL, we have a $(\gamma_{f}(\delta), 3\gamma_{c}(\delta))$-approximation algorithm for NFP-SFork problem, regardless of the forwarding graph's fork structure with any $\delta \geq 1$. When $\delta = 8.67$, we obtain a feasible solution for NFP-SFork, which is within a factor of 3.27 of the optimal solution of NFP-SFork.

\subsubsection{Extension: Robust NFP-SFork}\label{subsubsec:robustSfork}
To guarantee that there exists backup NF-enabled nodes after any NF-enabled node failure, we engage the robust fault-tolerance algorithm to ensure the availability of backup NF-enabled nodes, where the backup costs is also minimized at the level $\ell$ with $1\leq \ell \leq \gamma(\lambda)$ with $\lambda\in \Lambda$.

\cite{chechik2014robust} concluded the existence of $(1.5+7.5\alpha)$-approximation for the robust fault-tolerant facility location problem with $\alpha$ failed nodes. Hence, applying two-step parameterized path reduction to create SFC service path and backup paths after NF-enabled node failure into the approximation algorithm for robust fault-tolerant facility location problem, approximation algorithm exists for robust NFP problem with $\alpha$ NF-enabled node failure.

\subsection{Formulations for \zblue{Robust NF Provisioning with} NF Request}
We now present the mathematical formulations for the maximal reliable NF deployment problem based on the NF service failure probability. 
We first turn the non-linear objective $\min_{V^{F}_{P}}\max_{d_{st}\in \mathcal{D}}\min_{\eta_{st}\in \Pe_{st}}\Pi_{i\in V^{f}_{P}\cap \eta_{st}}\rho_i$ into its linearized counterpart
\begin{align}
\min_{V^{F}_{P}}\max_{\substack{f\in F_{st} \\ d_{st}\in \mathcal{D}}}\min_{\eta_{st}\in \Pe_{st}}\sum_{i\in V^{f}_{P}\cap \eta_{st}}\ln(\zblue{1+\rho_i})\label{fm:flProbObj}
\end{align}
by applying the $\ln(\cdot)$ function.

With Theorem~\ref{thm:evaluation}, the formulation presented below is the robust NF evaluation metric value of NF request with (\ref{fm:flProbObj}) as the objective:
\begin{align}
\min_{\lambda, x, y, z, \xi,h}\; &\lambda \nonumber \\
s.t.\; & \sum_{i\in V_P} h_i\leq N_f, f\in F \label{fm:nodeNumCon}\\
     & \lambda \geq \xi^{f}_{st}, \quad\quad\quad\quad\quad\quad\;\;\; f\in F, d_{st}\in \mathcal{D} \label{fm:lowProb}\\
     & \xi^{f}_{st} = \sum_{i\in V_P} ln(\zblue{1+\rho_{i}})y^{if}_{st},\quad\; f\in F, d_{st}\in \mathcal{D} \label{fm:cumProb}\\
     & y^{if}_{st} \geq z^{f}_{i} + \delta^{i}_{\eta_{st}}x_{\eta_{st}}+\gamma^{f}_{st} -2,\nonumber\\
     & \quad\quad\quad\;\; f\in F, d_{st}\in \mathcal{D}, \eta_{st}\in \Pe_{st}, i \in V_P \label{fm:funServ}\\
     & y^{if}_{st}\leq z^{f}_{i},\quad\quad\quad\quad f\in F, i\in V_P\label{fm:yUp1}\\
     & y^{if}_{st}\leq \delta^{i}_{\eta_{st}}x_{\eta_{st}}, \; f\in F, i\in V_P, d_{st}\in \mathcal{D},\nonumber\\
     &\qquad\qquad\qquad\qquad\qquad\qquad \eta_{st}\in \Pe_{st}\label{fm:yUp2}\\
     & y^{if}_{st}\leq \gamma^{f}_{st},\quad\quad\quad\;\; f\in F, i\in V_P, d_{st}\in \mathcal{D} \label{fm:yUp3}\\
     & h_{i} \geq z^{f}_{i}, \quad\quad\quad\quad\quad\quad\quad\quad f\in F, i\in V_P \label{fm:phyNdUsg}\\
     & \sum_{\eta_{st}\in \Pe_{st}} x_{\eta_{st}} =1, \quad\quad\quad\quad\quad\;\;  d_{st}\in \mathcal{D} \label{fm:pthSel}\\
     & \lambda,\xi^{f}_{st} \geq 0, z^{f}_{i}, y^{if}_{st}, h_{i}, x_{\eta_{st}} \in \{0,1\}, \eta_{st}\in \Pe_{st},\nonumber\\
     & \qquad (s,t)\in E_L, f\in F, d_{st}\in \mathcal{D}, i\in V_P \label{fm:region}
\end{align}
Constraint (\ref{fm:nodeNumCon}) enforces the upper bound for the number of nodes deployed with NFs. Constraint (\ref{fm:lowProb}) records the value of NF failure evaluation metric (linearized) among all demands for all NFs. Constraint (\ref{fm:cumProb}) captures 
the robust NF failure evaluation metric value (linearized, i.e., $\ln(\rho_i)$ as in constraint (\ref{fm:flProbObj})) of demand $d_{st}\in \mathcal{D}$ and $f\in F$. 
Based on Definition~\ref{def:nonChainPrb}, constraint (\ref{fm:funServ}) determines whether $f$ is deployed onto physical node $i$ for demand $d_{st}\in \mathcal{D}$, where (i) $z^{f}_{i}=1$ when $f$ is deployed onto physical node $i$; (ii) $\delta^{i}_{\eta_{st}}=1$ when node $i$ deployed with an NF is on a selected path $\eta_{st}$ for $d_{st}$; and (iii) $\gamma^{f}_{st}=1$ when $d_{st}$ requires NF $f$.
Constraints (\ref{fm:yUp1}) -- (\ref{fm:yUp3}) force variable $y^{if}_{st}$ to be 0 when any of the (i) to (iii) above is not satisfied.
Constraint (\ref{fm:phyNdUsg}) indicates whether physical node $i$ is deployed with any NFs. Constraint (\ref{fm:pthSel}) selects a single physical route for demand $d_{st}\in \mathcal{D}$. Constraint (\ref{fm:region}) provides feasible regions for all variables.

Note here that the variable $\lambda$ in constraint~(\ref{fm:lowProb}) records the value of the robust NF failure evaluation metric achieved by NF request through $\xi^{f}_{st}$. As the objective of the reformulation is to find the minimum $\lambda$, it also encourages evaluation metric value $\xi^{f}_{st}$ to be minimized. Therefore, the above reformulation solves the maximal reliable NF deployment problem.

We next present the formulation for SFC service reliability.

\subsection{Formulations for SFC Service Reliability}\label{subsubsec:sfcF}
Different from the non-chained NF failure probability, the \textbf{SFC failure probability} is $$1-\max_{\Gamma(F)}\min_{\substack{f\in F_{st} \\ d_{st}\in \mathcal{D}}}\max_{\eta_{st}\in \Pe_{st}}\left[1-\Pi_{i\in \Gamma(f)\cap p_{st}}\rho_i\right]/|F_{st}|!$$ with $d_{st}\in \mathcal{D}$.
\begin{proposition}\label{prop:sfcReq}
For requests with SFC, we have
$\max_{\Gamma(F)}\min_{\substack{f\in F_{st} \\ d_{st}\in \mathcal{D}}}\max_{p_{st}\in \Pe_{st}}\left[1-\Pi_{i\in \Gamma(f)\cap p_{st}}\rho_i\right]/|F^{*}_{st}|!$\\
$= 1-  \min_{\Gamma(F)}\max_{\substack{f\in F_{st} \\ d_{st}\in \mathcal{D}}}\min_{p_{st}\in \Pe_{st}}\Pi_{i\in \Gamma(f)\cap p_{st}}\rho_i /|F^{*}_{st}|!$
, where $F^{*}_{st}$ represents the requested NFs of $d^{*}_{st}=\arg\min_{d_{st}\in \mathcal{D}, f\in F_{st}}\left[\Pi_{i\in \Gamma(f)\cap p_{st}}\rho_i\right]$.
\end{proposition}

We introduce here an auxiliary variable $\omega_{st}$ which indicates whether $d_{st}\in \mathcal{D}$ is selected as the $d^{*}_{st}$.
By replacing routings from undirected to directed path set (i.e., $\eta_{st}\rightarrow p_{st}$) in constraints (\ref{fm:funServ}), (\ref{fm:yUp2}), (\ref{fm:pthSel}),  (\ref{fm:region}), we present the formulation for the robust SFC provisioning as follows.
\begin{align}
\min_{\lambda,\xi,\omega, \beta,y,x,z} & \lambda\nonumber\\
s.t.\; & \lambda \geq \omega_{st}, \quad\quad\quad\quad\quad\quad\quad\quad\quad\;\; d_{st}\in \mathcal{D} \label{fm:lowProbSFC}\\
     &\omega_{st} \geq \xi^{f}_{st} - \ln|F_{st}|!, \quad\quad f\in F, d_{st}\in \mathcal{D} \label{fm:stLowProbSFC}\\
     &\sum_{d_{st}\in \mathcal{D}}\beta_{st}=1\label{fm:selDmd}\\
     &\lambda \leq \omega_{st} + M(1-\beta_{st}),\quad\quad\quad\;\; d_{st}\in \mathcal{D} \label{fm:selDmdUpbd}\\
     &\lambda \geq \omega_{st} + M(\beta_{st} -1),\quad\quad\quad\;\; d_{st}\in \mathcal{D} \label{fm:selDmdLwbd}\\
     &\omega_{st}\geq 0, \beta_{st}\in \{0,1\}, d_{st}\in \mathcal{D}\label{fm:region2}\\
     &\text{Constraints (\ref{fm:nodeNumCon}) and (\ref{fm:cumProb})--(\ref{fm:region})}\nonumber
\end{align}

Constraint (\ref{fm:lowProbSFC}) is to guarantee the lower bound based on the $\mathcal{FP}$ (linearized). The newly introduced constraint (\ref{fm:stLowProbSFC}) is used to capture the corresponding SFC request $d_{st}\in \mathcal{D}$.
Constraint (\ref{fm:selDmd}) guarantees that exactly one demand $d_{st}\in \mathcal{D}$ should be selected as the $d^{*}_{st}$ which provides the $\mathcal{FP}(d_{st})$. Constraints (\ref{fm:selDmdUpbd}) and (\ref{fm:selDmdLwbd}) guarantee $\lambda = \omega_{st}$ for the selected $d^{*}_{st}$ (when $\beta_{st}=1$). 

\section{Simulation Results}\label{sec:computation}
We design our experiments for robust NF provisioning problems in two parts, (1) provisioning with non-chained NFs, and (2) provisioning with SFCs. 

\subsection{Experiment Design}\label{subsec:expDesign}
\subsubsection{Design for Robust NF Provisioning with Non-chained NFs}\label{subsubsec:nonChained}
We select NSF network as the physical network illustrated in Fig.~\ref{fig:nsf}, which has 14 nodes and 21 links.
\begin{figure}[b]
\centering
\includegraphics[scale=0.32]{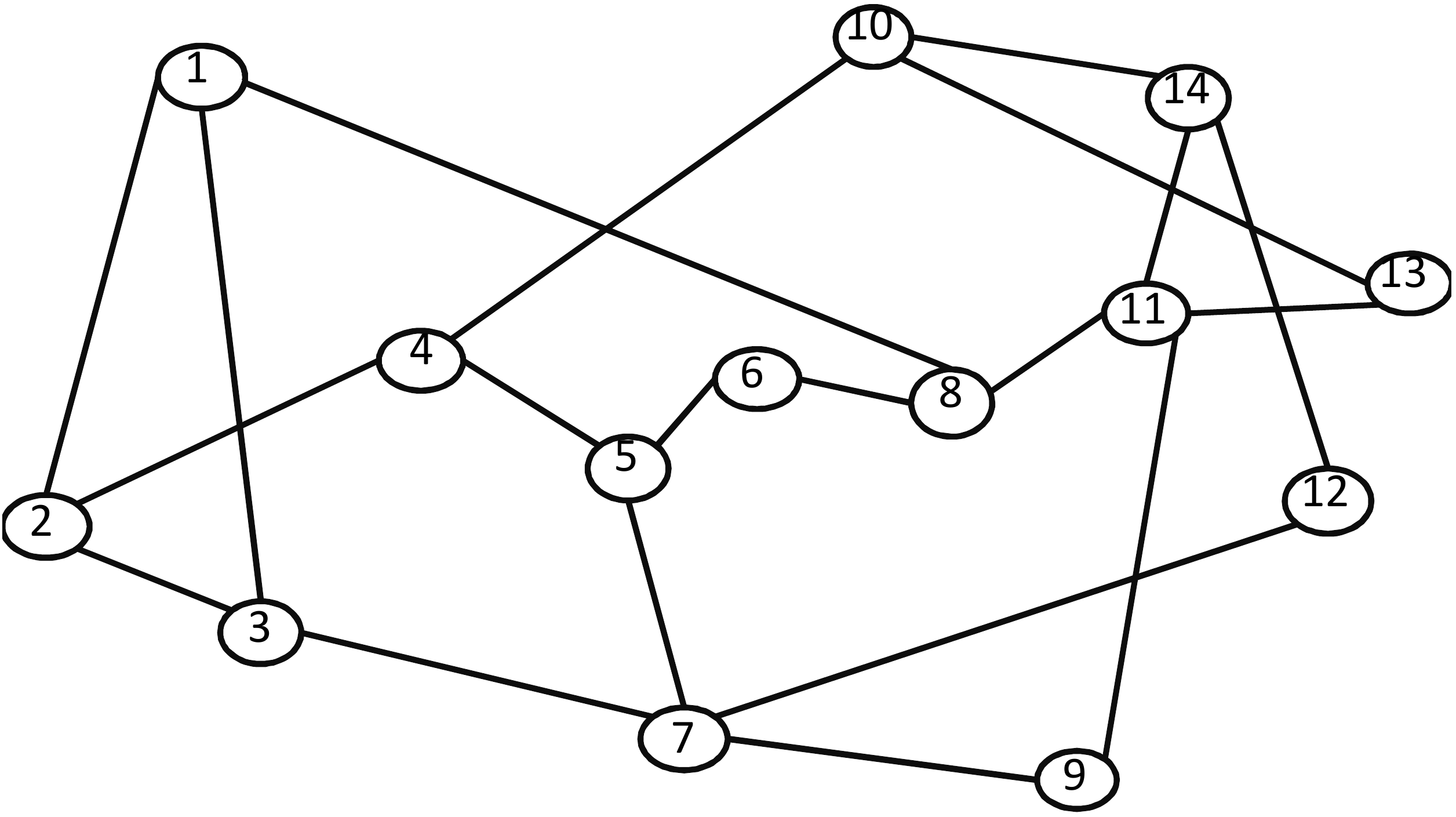}
\caption{NSF}
\label{fig:nsf}
\end{figure}
NF requests are based on node pairs whose mappings onto physical nodes are known a priori. Six pairs of NF requests are constructed and listed as follows: (1,2), (1,4), (2,3), (3,5), (4,7), and (6,7).
NF requests for logical arcs/links are randomly assigned with up to three NFs.

We consider that physical nodes are with random failure probabilities,  where the means of these probabilities are in the range of 1\% to 49\% and the variance is 0.001. For each of the failure probabilities, we generate 25 testing samples and report their average as the results.
For the simulations of the maximal reliable NF  deployment problem,  we first create testing cases  which restrict the number of NF-enable nodes to be 40\%, 50\%, and 60\% of the physical nodes.

Based on the settings above, two sets of testing cases are created.
The first testing cases for the maximal NF reliable deployment problem have (i) NSF as the physical network, (ii) demands with up to three randomly assigned NF requests, (iii) a given limitation on the number of NF deployed locations, and (iv) random node failure probability. The proposed setting is to verify that when the number of NF locations decreases, whether the NF service reliability also goes down corresponding. Meanwhile, when the node failure probability increases, whether the NF service reliability also decreases.

The second testing cases have (i) a fixed NF service reliability (90\%), and (ii) random physical node failure probability. The purpose of the setting is again to evaluate that with a fixed NF service reliability, whether extra NF-deployed nodes are required to fulfill the requirement of the service level when the node failure probability increases.
increases, whether extra NF-deployed nodes are required to fulfill the requirement of the service level.

\subsubsection{Design for Robust NF Provisioning with SFCs}\label{subsubsec:sfc}
We consider three different forwarding graphs, the single SFC (``1SFC''), rooted fork (``rFork''), and branched fork (``bFork''), and illustrate them in Fig~\ref{fig:frdGrp} to test the proposed robust evaluation metrics and approaches to calibrate their survivable probability.
\begin{figure}[!ht]
    \centering
    \includegraphics[width=6cm]{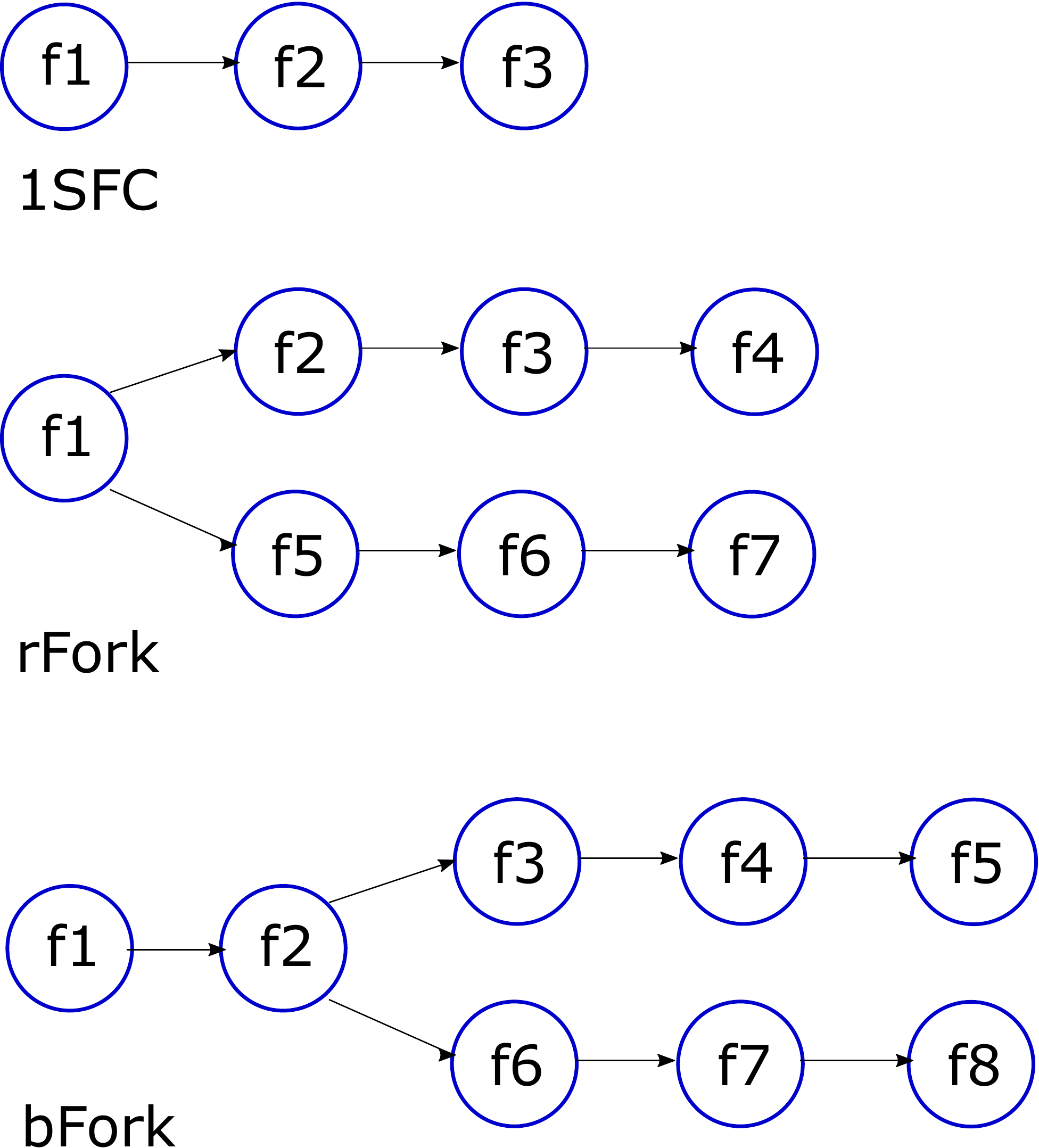}
    \caption{Forwarding graphs}
    \label{fig:frdGrp}
\end{figure}
We take CORONET network, illustrated in Fig.~\ref{fig:coronet}, as the physical network which has 75 nodes, 99 links, and an average nodal degree of 2.6.
\begin{figure}
    \centering
    \includegraphics[scale=0.4]{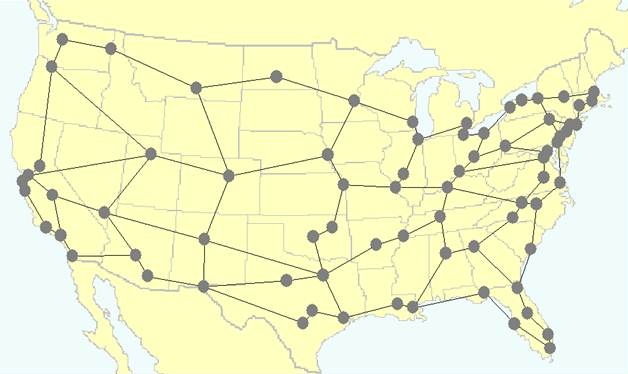}
    \caption{CORONET CONUS Network}
    \label{fig:coronet}
\end{figure}
We generate two sets of 6 demands and 10 demands randomly. For 1SFC as the forwarding graph, all demands require the SFC; for rFork and bFork, we randomly assign half of the demands (3 demands or 5 demands) with a branch of a fork, and the rest of demands require the second branch of forks.
Therefore, we have six testing cases, which are with three forwarding graphs and two demand pair setting correspondingly.

We also consider NF-enabled nodes to have failure probability from 1\% to 50\% with the variance be 0.1\%. The experiment is designed to test with how many NF-enabled nodes, all demand pairs would have positive survivable probability and what the numerical values are. We report the robust survivability probability of all testing cases with different NF-enabled node failure probability.

\subsection{Computational Results}\label{subsec:results}
\subsubsection{Computational Results for Robust NF Provisioning with Non-chained NFs}\label{subsubsec:cr_nonchained}
The simulation results for the maximal NF reliable deployment problem are presented in Fig.~\ref{fig:nfReqRlb}.
\begin{figure}[t]
\centering
\includegraphics[scale=0.42]{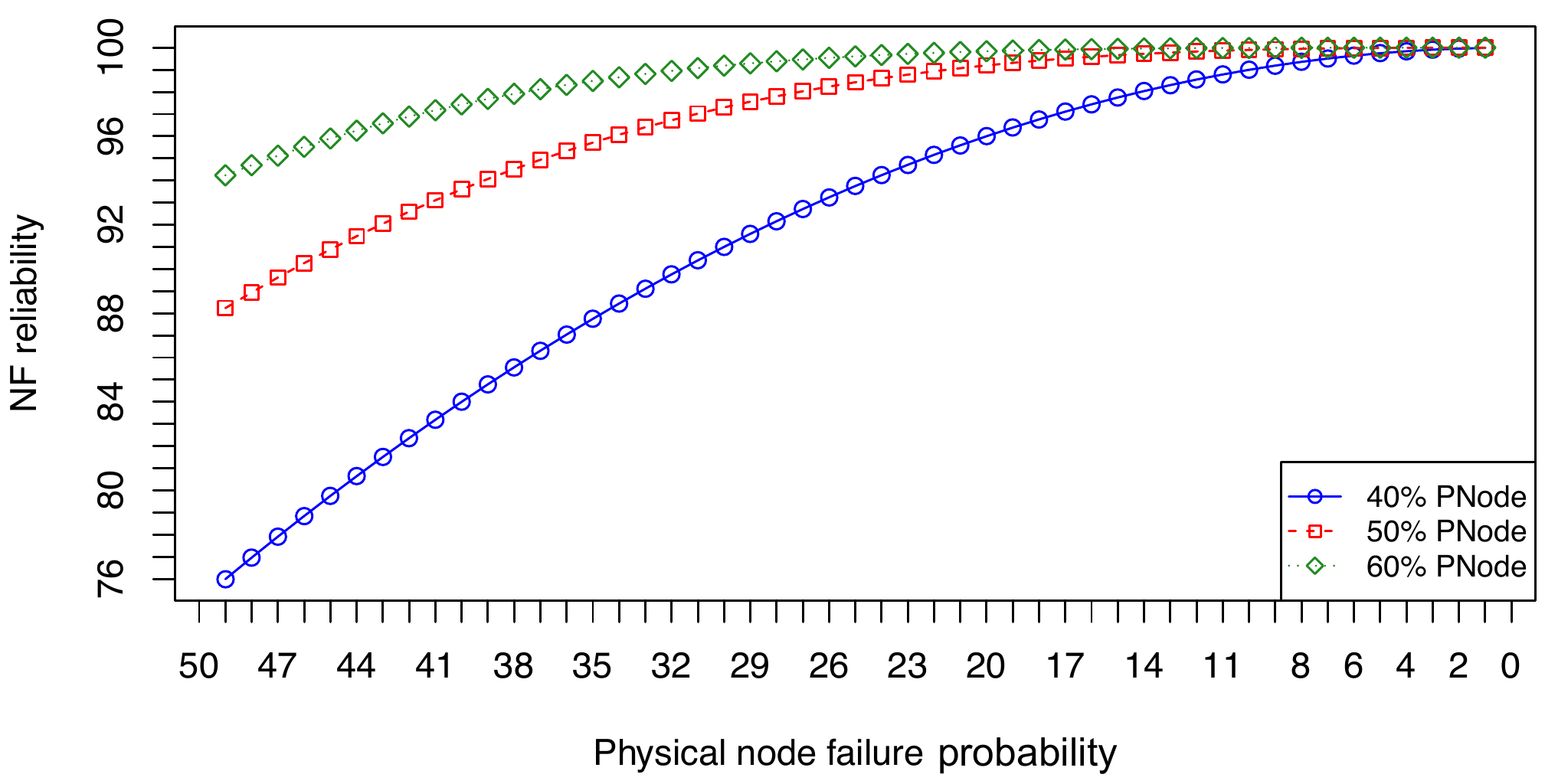}
\caption{NF service reliability}
\label{fig:nfReqRlb}
\end{figure}
The three lines in blue, red, and green colors represent the testing cases with 40\%, 50\%, and 60\% of NF-enabled physical nodes. The $x$-axis represents the physical node failure probability (in mean value) and the $y$-axis denotes the NF service reliability (in percentage). Each plotted node/dot in the figure presents the average NF service reliability for all testing samples. With up to 50\% failure probability of the NF-enabled nodes, the NF reliability reaches 75\%. When the number of NF-enabled nodes increases, the NF reliability increases to 87.5\%, and 93.7\%, respectively. We confirm our analysis that with the limitation on the number of NF-enabled nodes, the NF service reliability increases when physical node failure probability decreases. Also,  given the same physical node failure probabilities, we observe that when the number of NF-enabled nodes (in terms of the mean value) decreases, the reliability of the NF service decreases as well.

Figure~\ref{fig:nfRblnfDp} illustrates the number of NFs deployed to reach the required level of the NF service reliability (based on the maximal number of NF-enabled nodes in the testing cases) with single NF and multiple NFs (in our testing cases, three required NFs) in each demand. To reach the fixed (90\%) NF service reliability, the number of physical nodes deployed with NFs is only doubled when the number of required NFs for each demand goes from one to three even with high failure probability (10 -- 50\%) on physical nodes. The figure demonstrates a clear pattern between the number of nodes deployed with NFs and the NF service reliability.

In the simulation results, we observe that the NF service reliability is higher with more physical nodes deployed with the required NFs, and obviously, a lower average node failure probability leads to a higher NF service reliability under the failure(s) of physical nodes. 
\begin{figure}[t]
\centering
\includegraphics[scale=0.42]{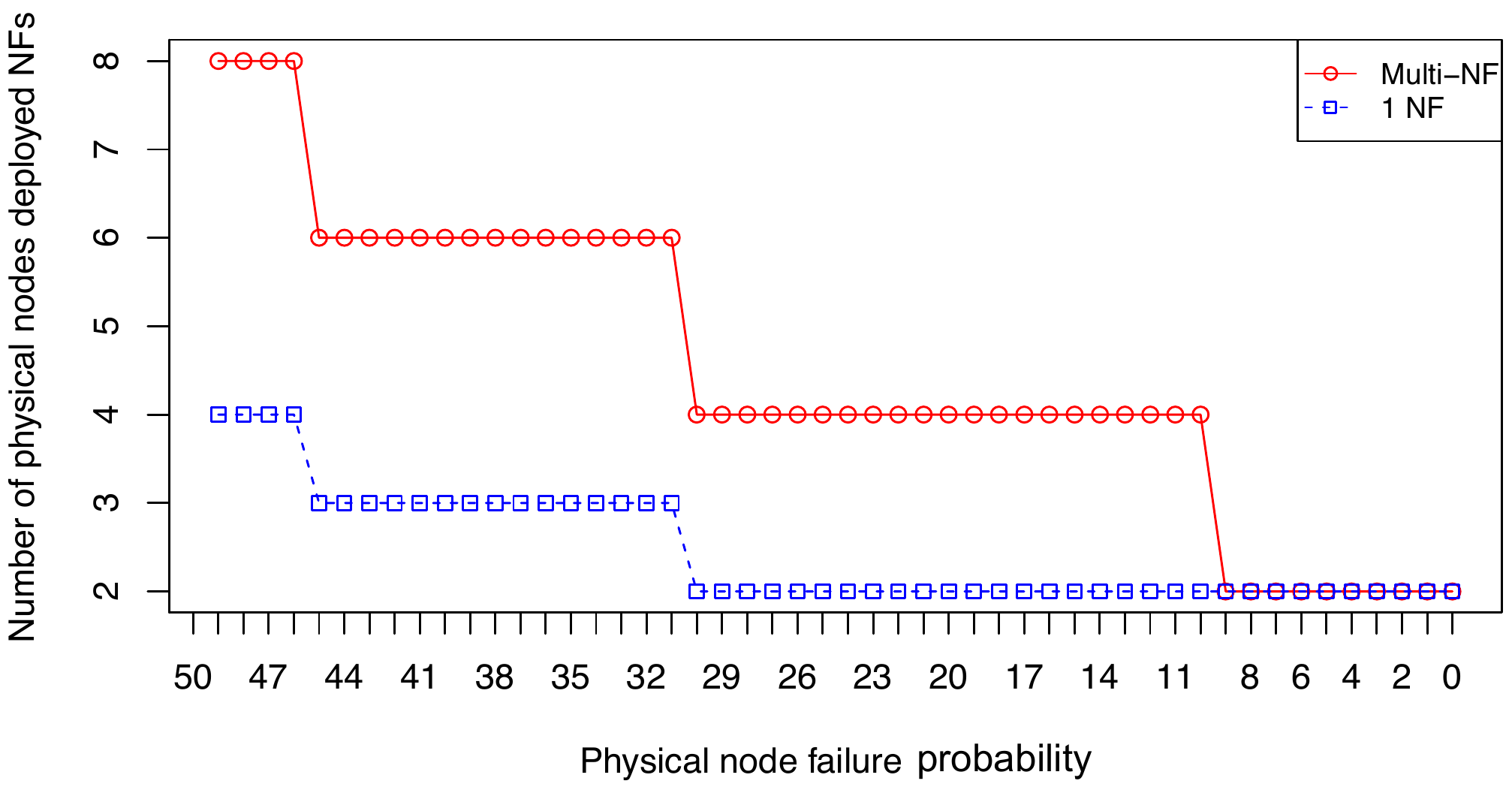}
\caption{NF service reliability vs. NF deployment}
\label{fig:nfRblnfDp}
\end{figure}
The observations on these simulations are as expected and demonstrate the relationship between the number of NF-deployed nodes (cost-related restriction) and NF service reliability (service level).

\subsubsection{Computational Results for Robust NF Provisioning with SFCs}\label{subsubsec:cr_sfc}
Following the experiment design in Section~\ref{subsubsec:sfc}, we test six SFC requests and ten SFC requests cases separately. 

Different from the robust survivable NF provisioning problem with non-chained NFs, the SFC requests need to visit required NFs in order. Therefore, the NF instance deployment on NF-enabled nodes is more restricted. We proceed with our testing in the larger-scale physical network, the CORONET network, to identify (1) how many NF-enabled nodes are needed, and (2) what are the corresponding survivable probability among all SFC requests to guarantee that all SFC requests are fulfilled with three types of SFC forwarding graphs. 

We first report the NF-enabled nodes required to support all demands with different NF forwarding graphs, and have positive survivability, in Figs.~\ref{fig:dm6nfNds} and \ref{fig:dm10nfNds} for six and ten SFC requests.

\begin{figure}[!ht]
    \centering
    \includegraphics[scale=0.4]{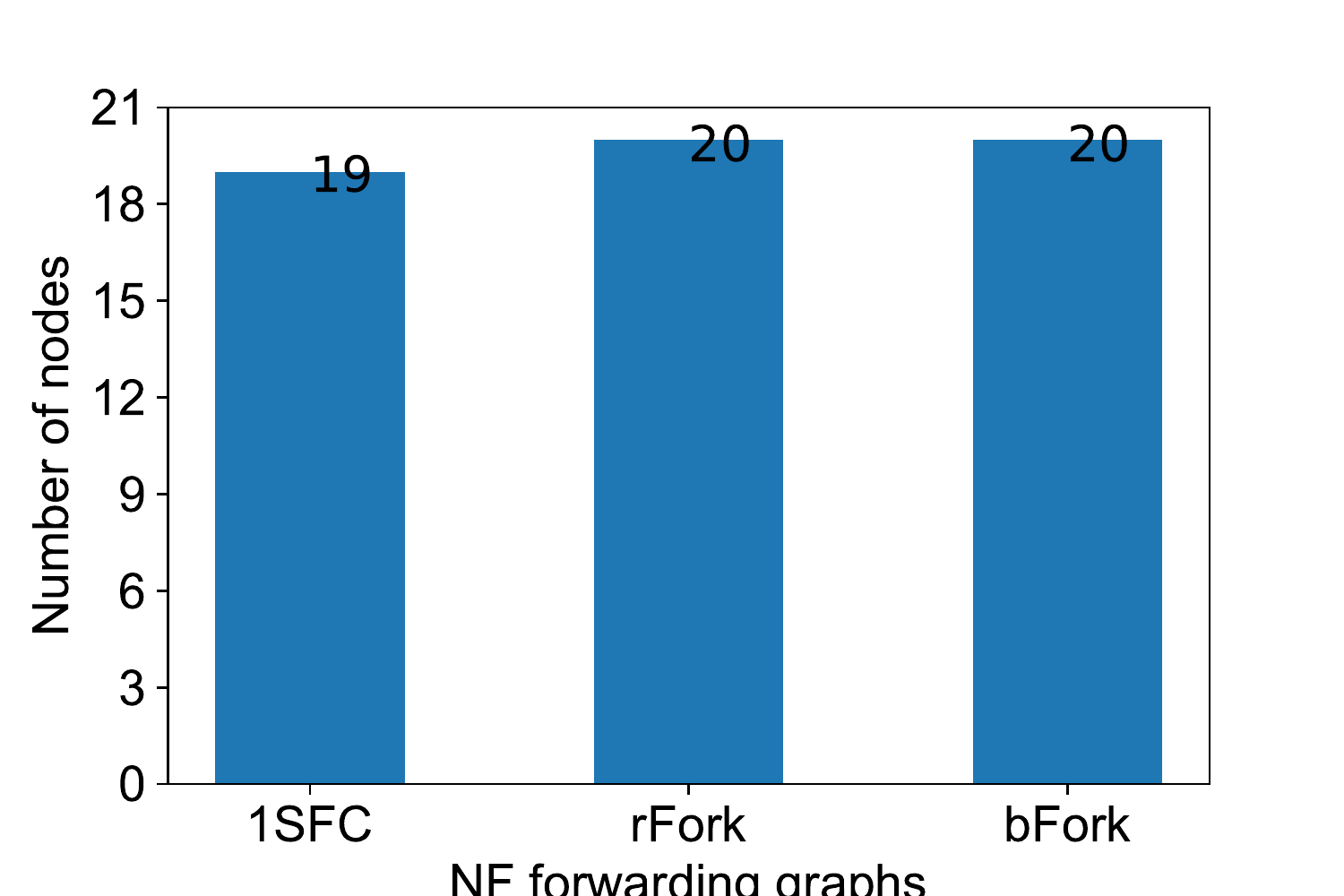}
    \caption{Number of NF-enabled nodes for six demand pairs}
    \label{fig:dm6nfNds}
\end{figure}
\begin{figure}[!ht]
    \centering
    \includegraphics[scale=0.4]{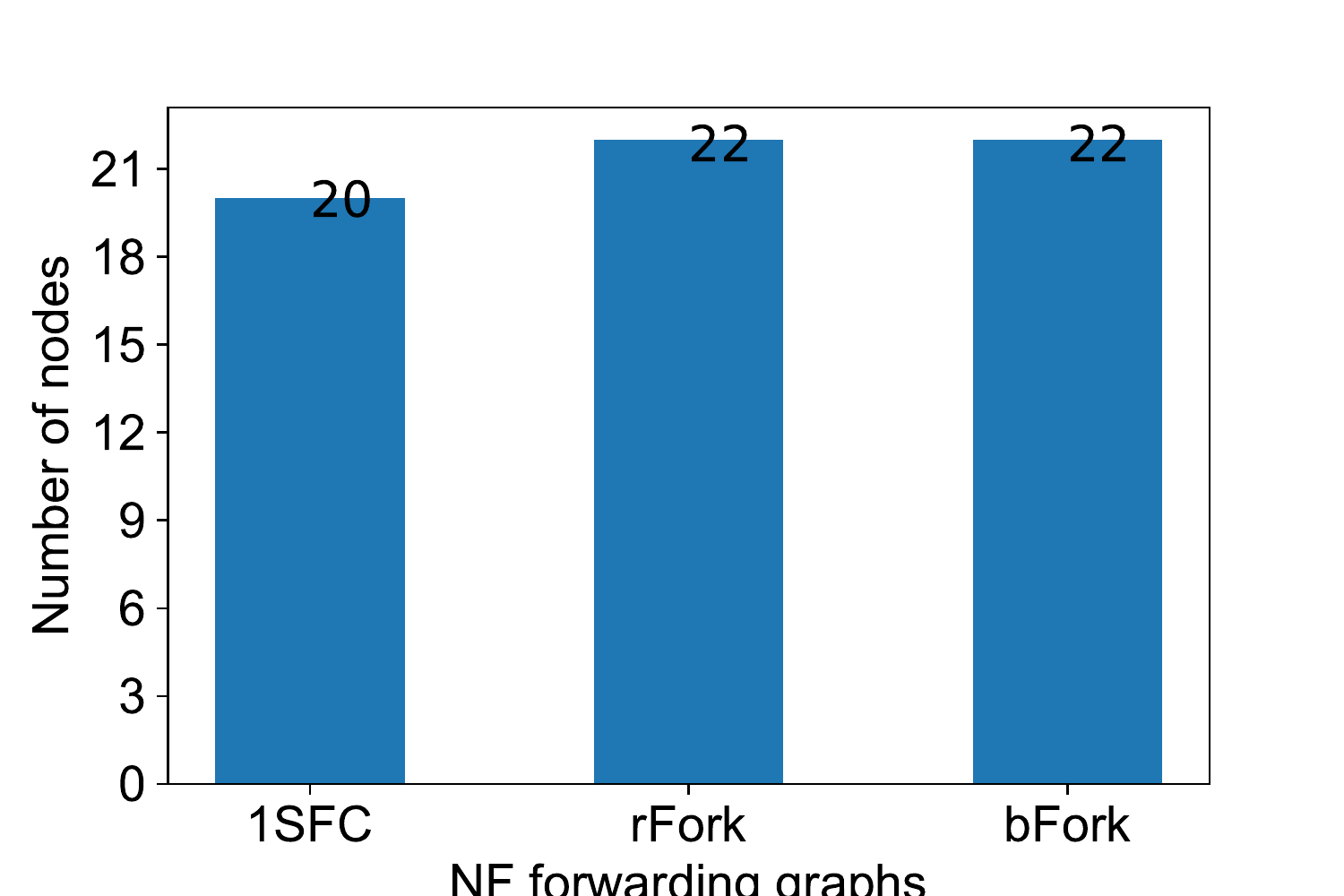}
    \caption{Number of NF-enabled nodes for ten demand pairs}
    \label{fig:dm10nfNds}
\end{figure}

The results meet our expectation that with the same requests, more NF-enabled nodes are required to guarantee that all SFC requests have a positive survivable probability. The more requests are added, the requirement of NF-enabled nodes increased. Note here that we did not observe too much difference on the number of NF-enabled nodes needed with rFork and bFork as the NF forwarding graphs for both six and ten SFC requests. 

Next, we present the robust survivable probability of all testing cases with NF-enabled nodes failure probability from 5\% to 50\% (with a fixed 5\% gap) for six and ten demands in Table\ref{tbl:svProbDm6_Dm10}. We let ``FP'', ``6D-1SFC'', ``6D-rFork'',  ``6D-bFork'', ``10D-1SFC'', ``10D-rFork'', ``10D-bFork'' represent the failure probability, survivable probability of six and ten demands with 1SFC, rFork and bFork as the forwarding graphs, respectively. Computational results show a very clear pattern that the higher failure probability, the lower the survivable probability for all SFC requests. Compared with rFork and bFork, 1SFC as forwarding graph has much higher survivable probability; and with bFork as the forwarding graph, the survivable probability is low. The highest we could reach is around 6.28\%.
\begin{table}[!ht]
\centering
\begin{tabular}{p{0.2cm}p{0.9cm}p{0.9cm}p{0.9cm}p{0.9cm}p{0.9cm}p{0.9cm}}
\hline
  FP &   6D-1SFC &  6D-rFork &  6D-bFork &   10D-1SFC &  10D-rFork &  10D-bFork \\
\midrule
          5 &  94.1613 &  31.3871 &  6.2774 &  95.9597 &   29.8925 &   5.9785 \\
         10 &  93.4533 &  31.1511 &  6.2302 &  90.9091 &   28.3192 &   5.6638 \\
         15 &  92.2734 &  30.7578 &  6.1516 &  85.8586 &   26.7459 &   5.3492 \\
         20 &  90.6216 &  30.2072 &  6.0414 &  80.8081 &   25.1726 &   5.0345 \\
         25 &  88.4976 &  29.4992 &  5.8998 &  75.7575 &   23.5993 &   4.7199 \\
         30 &  85.9017 &  28.6339 &  5.7268 &  70.7073 &   22.0261 &   4.4052 \\
         35 &  82.8336 &  27.6112 &  5.5222 &  65.6568 &   20.4528 &   4.0906 \\
         40 &  79.2939 &  26.4313 &  5.2863 &  60.6062 &   18.8795 &   3.7759 \\
         45 &  75.2820 &  25.0940 &  5.0188 &  55.5557 &   17.3062 &   3.4612 \\
         50 &  70.7979 &  23.5993 &  4.7199 &  50.5051 &   15.7329 &   3.1466 \\
\hline
\end{tabular}
\vspace{2pt}
\caption{Robust survivable probability of SFC requests}
\label{tbl:svProbDm6_Dm10}
\end{table}

To have a finer granularity of the robust failure probability and observe the patterns of changes for robust survivable probability, we plot the robust failure probability in terms of NF-enabled failure probability from 50\% to 1\% (with 1\% gap). With six SFC requests, the survivable probability is curved and convex. The trend line of the survivable probability has a smoother increase when thee failure probability is lower. With ten SFC requests, we observe that the changes to survivable probability are more linear in terms of the failure probability.

\begin{figure}[!ht]
    \centering
    \includegraphics[width=8.5cm]{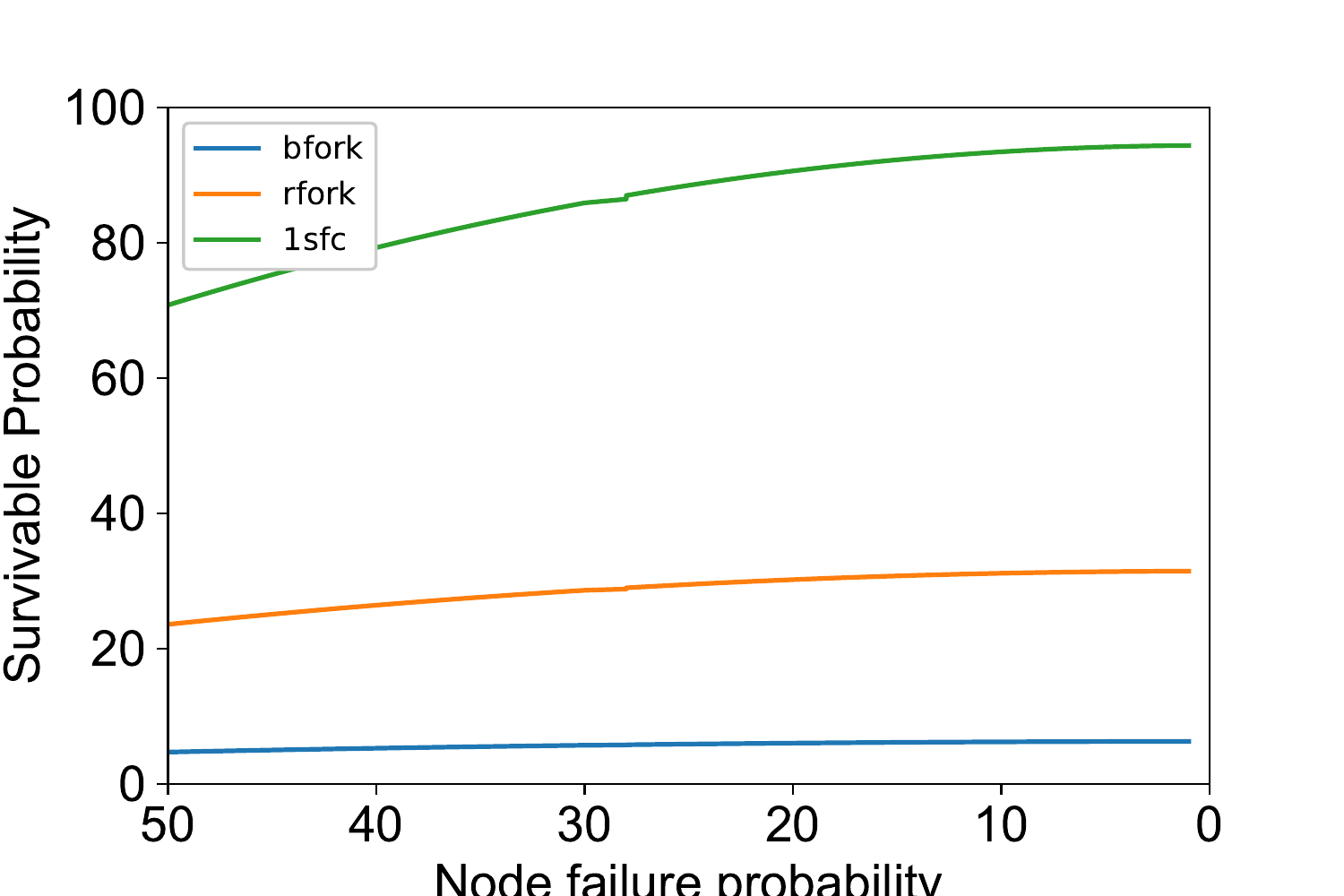}
    \caption{Six demand robust survivable probability}
    \label{fig:dm6RbSvPrb}
\end{figure}
\begin{figure}[!ht]
    \centering
    \includegraphics[width=8.5cm]{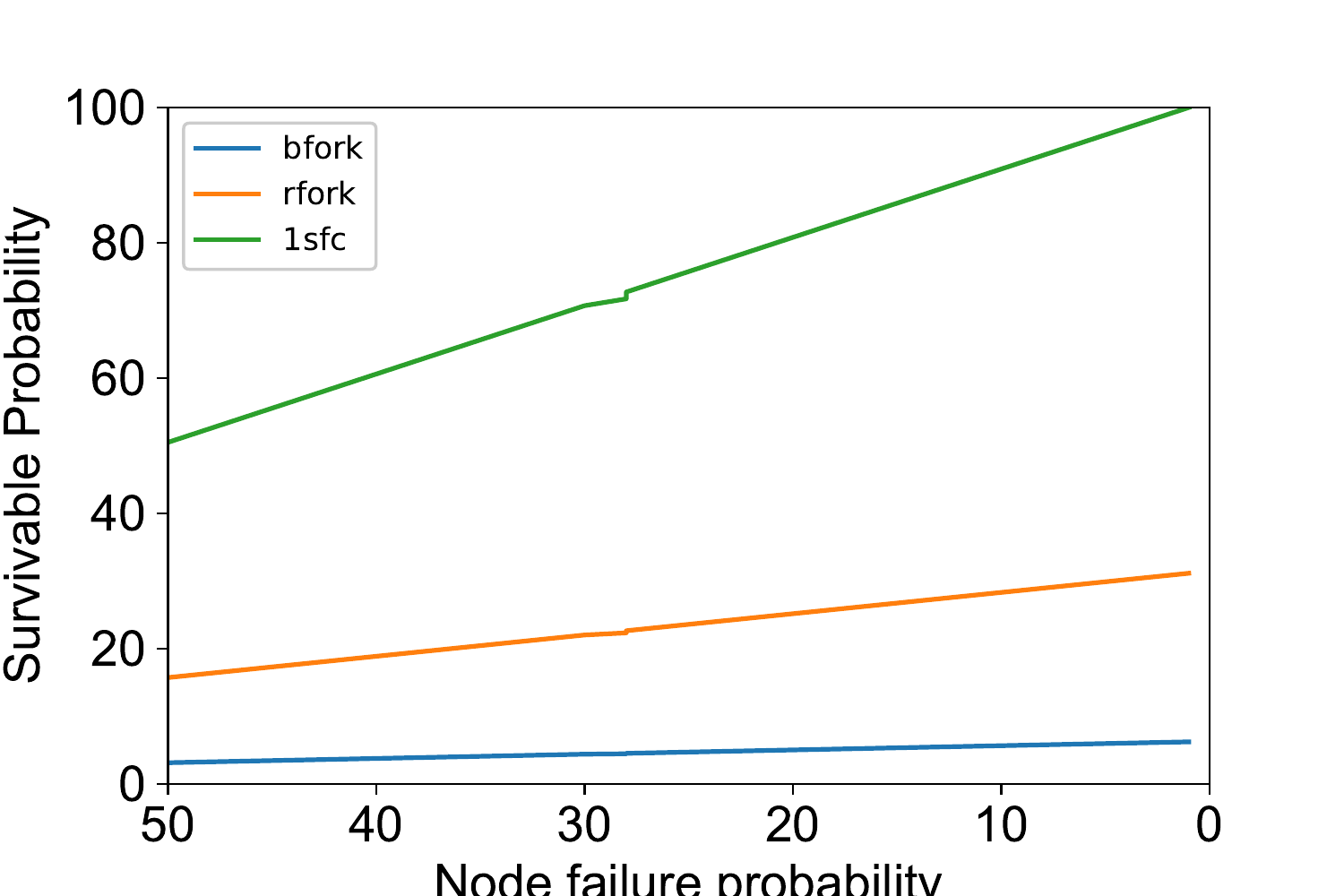}
    \caption{Ten demand robust survivable probability}
    \label{fig:dm10RbSvPrb}
\end{figure}

\section{Conclusion}\label{sec:conclusion}
In this study, we address the practical challenges in NFV 5G implementation. We propose a new robust evaluation metric that quantifies the minimal reliability among all NFs for all demands considering the random NF-enabled node failure. We study three sets of problems in the robust NF provisioning, i.e., the SFC $s-t$ path problem, NFP-SFork, and NFP with the general NF forwarding graph. We introduce an auxiliary/augmented network layer and we develop pseudo-polynomial algorithms to solve the robust NF and SFC $s-t$ path problems. We present approximation algorithms for robust NFV with the SFC-Fork as the NF forwarding graph and adopt a two-step parameterized path reduction technique, which can serve in multiple types of approximation algorithms when the underlying network has the branching structure. Furthermore, we propose exact solution approaches via MILP with general forwarding graph structures. Computational results show that our proposed solution approaches are capable of managing robust NFP with non-chained NF and SFC requests in both small and large-size national-wide physical networks. 

In further research, we would like to consider physical node capacity and NF deployment costs. We are also interested in evaluating the costs to introduce more NF-enabled nodes in the physical network. Another line of investigation is on the scenarios of shared risk group failure(s) and physical link failure(s) and their impacts on NF service reliability. Last, but not least, another research direction is to relax the assumptions on independent node failures, the correlation among NF-enabled node failures, and study their impacts on NF service reliability.

\section*{Appendix: Proof of Theorem 10}\label{sec:bifactSFCproof} 
\begin{lemma}\label{lm:forkCost}
Given a feasible solution $\varphi_{\text{1fl}}$ of the 1-level facility location problem based on NFP-SFork reduction its corresponding NFP-SFork solution $\varphi_{\text{sfc}}$, we have $F(\varphi_{\text{sfc}})+C(\varphi_{\text{sfc}}) \leq F(\varphi_{\text{1fl}}) +C(\varphi_{\text{1fl}})$.
\end{lemma}
\begin{proof}
Given a feasible solution $\varphi_{\text{1fl}}$ of the 1-level facility location problem, we have
\begin{align}
&C(\varphi_{\text{1fl}})=\sum_{d\in D}\left[C(d,t)+C(p(t,i_{f_{b+1}})+C(p(i_{f_{b+1}},i_{f_1}))\right] +\nonumber\\
&\qquad\qquad \sum_{t=1,\cdots, |D|}F(p(t,i_{f_{b+1}})),\nonumber\\
&F(\varphi_{\text{1fl}})=\sum_{i_{f_{b+1}}\in X_{b+1}} F(p(i_{f_{b+1}}, i_{f_1})).
\end{align}
where $X_{b+1}$ contains NF-enabled nodes selected by paths $\cup_{t\in 1,\cdots, |D|}p(t,i_{f_{b+1}})$.
For the corresponding NFP-SFork solution $\varphi_{\text{sfc}}$, we have
\begin{align}
&O(\varphi_{\text{sfc}}) \leq \sum_{i_{f_{b+1}}\in X_{b+1}} F(p(i_{f_{b+1}}, i_{f_1}))
,\nonumber\\
&C(\varphi_{\text{sfc}})=\sum_{d\in D}\left[C(d,t)+C(p(t,i_{f_{b+1}}))+C(p(i_{f_{b+1}},i_{f_1}))\right]\nonumber\\
&\qquad\qquad +\sum_{t=1, \cdots, |D|}F(p(t,i_{f_{b+1}})).
\end{align}
Hence, the conclusion holds.
\end{proof}

\begin{lemma}\label{lm:forkForest}
Given a solution of NFP-SFork, $\chi_{\text{sfc}}$, then there exists another solution $\psi_{\text{sfc}}$ such that\\
(1) in solution $\chi_{\text{sfc}}$ and given $\lambda\in \Lambda$, paths $\rho_{1}=(i_{f_{1(\lambda)}}, i_{f_{2(\lambda)}},\cdots, i_{f_{\gamma(\lambda)}})$ and $\rho_{2}=(i'_{f_{1(\lambda)}}, i'_{f_{2(\lambda)}},\cdots, i'_{f_{\gamma(\lambda)}})$ with $\ell(\lambda)$ for some $i_{\ell(\lambda)}=i'_{\ell(\lambda)}$, then in solution $\psi_{\text{sfc}}$,  $i_{j(\lambda)}=i'_{j(\lambda)}$ with $1\leq j(\lambda)\leq \ell(\lambda)$, and \\
(2) in solution  $\chi_{\text{sfc}}$ with given $\lambda, \lambda'\in \Lambda$, paths $\rho_{1}=(i_{f_{1(\lambda)}}, i_{f_{2(\lambda)}},\cdots, i_{f_{|\lambda|}})$ and $\rho_{2}=(i'_{f_{1(\lambda')}}, i'_{f_{2(\lambda')}},\cdots, i'_{f_{|\lambda'|}})$ with $i_{f_{\ell(\lambda)}}=i'_{f_{\ell'(\lambda')}}$ for some $\ell(\lambda)$, $\ell'(\lambda')$, then, in solution $\psi_{\text{sfc}}$, $i_{j}=i'_{j}$, for all $1\leq j \leq \ell(\lambda)$; and\\
(3) $O(\psi_{\text{sfc}})\leq O(\chi_{\text{sfc}})$ and $C_{|\lambda|}(\psi_{\text{sfc}})= C_{|\lambda|}(\chi_{\text{sfc}})$, and
$\sum_{1\leq j\leq |\lambda|-1}C_{j}(\psi_{\text{sfc}})\leq \sum_{1\leq j\leq |\lambda|-1}C_{j}(\chi_{\text{sfc}})$
\end{lemma}
\begin{proof}
Proof of Claim 1: Given a solution $\chi_{\text{sfc}}$ of NFP-SFork, we show how to obtain $\psi_{\text{sfc}}$ based on $\chi'_{\text{sfc}}$, a feasible solution for NFP-SFork.
In $\chi_{\text{sfc}}$, there exist paths $\rho_{1}=(i_{f_{1(\lambda)}}, i_{f_{2(\lambda)}},\cdots, i_{f_{\gamma(\lambda)}})$ and $\rho_{2}=(i'_{f_{1(\lambda)}}, i'_{f_{2(\lambda)}},\cdots, i'_{f_{\gamma(\lambda)}})$ with $\ell_{\lambda}$ for some $i_{f_{\ell(\lambda)}}=i'_{f_{\ell(\lambda)}}$, but for all $1\leq h(\lambda)\leq \ell_{\lambda}$, $\rho_{1}(f_{h(\lambda)})\neq \rho_{2}(f_{h(\lambda)})$.
We let $\rho_{2}(f_{\ell(\lambda)})=\rho_{1}(f_{\ell(\lambda)})$, and only deploy VNF instances on $\rho_{1}$ without the deployment on $\rho_{2}(f_{\ell(\lambda)})$,
which is still be a feasible solution for NFP-SFork. In $\psi_{\text{sfc}}$, $\rho_{1}=(i_{f_{1(\lambda)}}, i_{f_{2(\lambda)}},\cdots, i_{f_{\ell(\lambda)}}, \cdots, i_{f_{\gamma(\lambda)}})$ and $\rho_{2}=(i'_{f_{1(\lambda)}}, i'_{f_{2(\lambda)}},\cdots, i_{f_{\ell(\lambda)}}, \cdots, i'_{f_{\gamma(\lambda)}})$ with $\ell_{\lambda}$.

Proof of Claim 2: Similar to thee proof above, with solution $\chi_{\text{sfc}}$, we alter the SFC service path $\rho_{2}$ as 
\begin{align*}
\rho_{2}=(i_{f_{1(\lambda)}}, i_{f_{2(\lambda)}},\cdots, i_{f_{\ell(\lambda)}},\cdots, i'_{f_{|\lambda'|}})),
\end{align*}
which still provides a feasible solution for NFP-SFork. Solution $\psi_{\text{sfc}}$ takes $\rho_1$ and alters $\rho_2$ as the service paths.  \\ 
With Claims 1 and  2, Claim 3 holds.
\end{proof}

\begin{lemma}\label{lm:fork13app}
$\alpha F(\psi_{\text{1fl}})+ \beta C(\psi_{\text{1fl}}) \leq  \alpha O(\psi_{\text{sfc}})+ 3\beta C(\psi_{\text{sfc}})$
\end{lemma}
\begin{proof}
Given a solution $\chi_{\text{sfc}}$, we construct $\psi_{\text{sfc}}$. Without loss of generality, we assume that $\psi_{\text{sfc}}$ satisfies Lemma~\ref{lm:forest}. Thus, we have $O(\psi_{\text{sfc}}) \leq O(\chi_{\text{sfc}})$.

We let $X_{1}\in N(f_1)$ and $X_{b+1}\in N(f_{b+1})$ be two NF-enabled node sets to be deployed with NF instances. With Lemma~\ref{lm:forkForest}, we let $D(\mu)\subset D$ be the request set connecting to $\mu \in X_{b+1}$, where $p(\mu)=\arg\min_{d\in D(\mu)}C(p(d,i_{f_{b+1}}))$. We let $I(\nu)\in N(f_{b+1})$ be NF $b$-enabled nodes connecting to $\nu\in X_{1}$ following the solution $\chi_{\text{sfc}}$. We let 
\begin{align*}
p(\nu)=\arg\min_{i_{f_{b+1}\in I(f_{b+1}})}C(p(i_{f_{b+1}},i_{f_1})),
\end{align*} 
$\mathcal{P}_{1}$ be the SFC service subpath set in $\psi_{\text{sfc}}$ connecting NF 1 enabled nodes to the NF $b$ enabled nodes; and $\mathcal{P}_{2}$ be the SFC service subpath set in $\psi_{\text{sfc}}$ connecting NF ($b+1$) enabled node to NF $|\gamma(\lambda)|$ enabled node.
We create a new solution, $\varsigma_{\text{sfc}}$, for NFP-SFork where all demands in $D(\mu)$ are connected to $p(\mu)$, and $j\in I(f_{b+1})(\nu)$ are connected to $p(\nu)$.

Hence, we decompose the connection cost of NFP-SFork into four parts:
\begin{align*}
&C(\psi_{\text{sfc}})=C_{\gamma}(\psi_{\text{sfc}}) + C_{b}(\psi_{\text{sfc}})+ C_{\mathcal{P}_{1}}(\psi_{\text{sfc}})+C_{\mathcal{P}_{2}}(\psi_{\text{sfc}})\\
&\text{where}\\
&C_{\gamma}(\psi_{\text{sfc}})=\sum_{(d,i_{f_{\gamma(\lambda)}})\in \bigcup_{d\in D}p(d,i_{f_{b+1}})}  C(d,i_{f_{\gamma(\lambda)}}),\\
       &C_b(\psi_{\text{sfc}}) =\sum_{(i_{f_{b+1}},i_{f_{b}})\in \bigcup_{i_{f_{b+1}}\in X_{b+1}}p(i_{f_{b+1}},i_{f_{1}})}C(i_{f_{b+1}},i_{f_{b}}),\\
       &C_{\mathcal{P}_{1}}(\psi_{\text{sfc}})=\sum_{i_{f_{b+1}}\in X_{b+1}}C(p(i_{f_{b+1}}, i_{f_1})),\\
       &C_{\mathcal{P}_{2}}(\psi_{\text{sfc}})=\sum_{t= 1,\cdots, |D|}C(p(t,i_{f_{b+1}})).
\end{align*}
For $d\in D(\mu)$ with $\mu \in X_{b}$,
\begin{align*}
    C_{\mathcal{P}_{2}}(\varsigma_{\text{sfc}})+C_{\gamma}(\varsigma_{\text{sfc}})
\leq 3C_{\mathcal{P}_{2}}(\psi_{\text{sfc}})+C_{\gamma}(\psi_{\text{sfc}})
\end{align*}
With the fact that given a demand $d\in D(\mu)$, $C(p(u),\varsigma_{\text{sfc}})+C((d,p(u)),\varsigma_{\text{sfc}}) \leq 3C(p(d),\psi_{\text{sfc}})+C((d,p(u)),\psi_{\text{sfc}})$
follows the triangle inequality (as illustrated in Fig~\ref{fig:tEq}).
\begin{figure}
\centering
\includegraphics[scale=0.32]{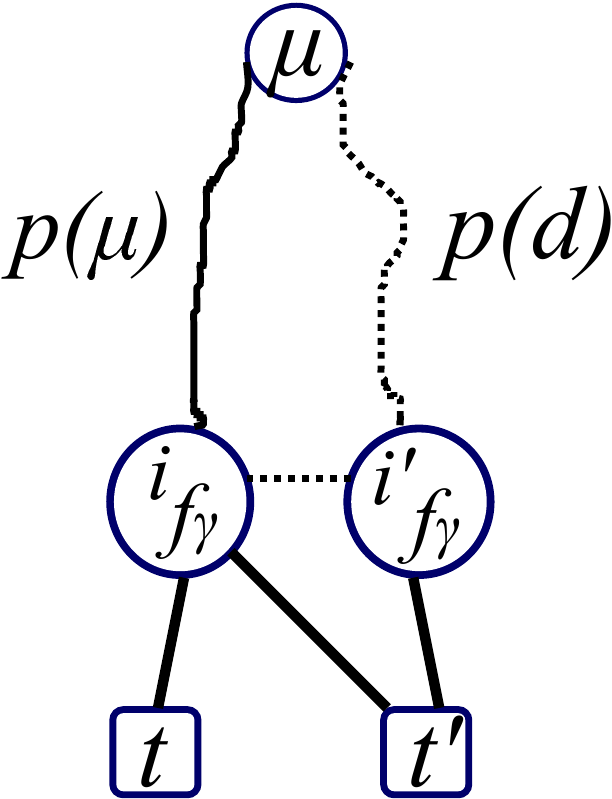}
\caption{Triangle inequality for connection costs}
\label{fig:tEq}
\end{figure}
Through the similar idea, we have
\begin{align*}
    C_{\mathcal{P}_{1}}(\varsigma_{\text{sfc}})+C_{(b,b+1)}(\varsigma_{\text{sfc}})
\leq 3C_{\mathcal{P}_{1}}(\psi_{\text{sfc}})+C_{(b,b+1)}(\psi_{\text{sfc}})
\end{align*}
We create the 1-level facility location solution, $\varsigma_{\text{1fl}}$, as follows:, all NF-enabled nodes on path $p(\nu,|I(\nu)|)$ and $p(j,|D_{\mu}|)$ with $j\in X_{b+1}$ are deployed with NFs, and $d$ is connected to path $p(j,|D_{\mu}|)$ where $d\in D_{\mu}$, and $p(j,|D(\mu)|)$ is connected to $p(\nu, |I(\nu)|)$, where $j\in I(\nu)$.
\begin{align}
&\alpha O(\varsigma_{\text{1fl}})+ \beta C(\varsigma_{\text{1fl}})\nonumber\\
=&\alpha \sum_{\nu\in X_1} O(p(\nu,|I(\nu)|))+ \beta \sum_{\nu\in X_{1}}\sum_{j\in I(\nu)}C(p(\nu, |I(\nu)|))\nonumber\\
&+\beta C_{b} + \beta \sum_{\mu\in X_{b+1}}\sum_{d\in D(\mu)}[C(p(d,\mu))+ O(p(d,\mu))]+\beta  C_{\gamma}\nonumber\\
\leq &\sum_{\nu\in X_1} \left[\alpha O(p(\nu,|I(\nu)|))+ \beta |I(\nu)|C(p(\nu, |I(\nu)|))\right]+\beta C_{b}\nonumber\\
&+ \sum_{\mu\in X_{b+1}}\beta |D(\mu)|[C(p(d,\mu))+ O(p(d,\mu))]+\beta  C_{\gamma}\nonumber\\
\leq &\sum_{\nu\in X_1}\left[\alpha O(p(\nu)) + \beta |I(\nu)|C(p(\nu))\right]+\beta C_{b}\nonumber\\
&+\sum_{\mu\in X_{b+1}}|D(\mu)|\left[\beta O(p(\mu))+ \beta C(p(\mu)) \right]+\beta C_{\gamma}\label{fm:pathConstruction}\\
\leq &\alpha O_{1}(\varsigma_{\text{sfc}})+\beta \left[C_{\mathcal{P}_{1}}(\varsigma_{\text{sfc}})
+ C_{b}(\varsigma_{\text{sfc}}) + C_{\mathcal{P}_{2}}(\varsigma_{\text{sfc}})\right.\nonumber\\
&\left.+ C_{\gamma}(\varsigma_{\text{sfc}})+O_{2}(\varsigma_{\text{sfc}})\right]\nonumber
\end{align}
where inequality (\ref{fm:pathConstruction}) holds with two-step parameterized path reduction.

Hence, the conclusion holds.
\end{proof}

\bibliographystyle{IEEEtran}
\bibliography{net19-req}

\begin{thebibliography}{10}
\providecommand{\url}[1]{#1}
\csname url@samestyle\endcsname
\providecommand{\newblock}{\relax}
\providecommand{\bibinfo}[2]{#2}
\providecommand{\BIBentrySTDinterwordspacing}{\spaceskip=0pt\relax}
\providecommand{\BIBentryALTinterwordstretchfactor}{4}
\providecommand{\BIBentryALTinterwordspacing}{\spaceskip=\fontdimen2\font plus
\BIBentryALTinterwordstretchfactor\fontdimen3\font minus
  \fontdimen4\font\relax}
\providecommand{\BIBforeignlanguage}[2]{{%
\expandafter\ifx\csname l@#1\endcsname\relax
\typeout{** WARNING: IEEEtran.bst: No hyphenation pattern has been}%
\typeout{** loaded for the language `#1'. Using the pattern for}%
\typeout{** the default language instead.}%
\else
\language=\csname l@#1\endcsname
\fi
#2}}
\providecommand{\BIBdecl}{\relax}
\BIBdecl

\bibitem{alliance20155g}
\BIBentryALTinterwordspacing
N.~Alliance, ``{5G} white paper,'' Mar. 2015. [Online]. Available:
  \url{https://www.ngmn.org/5g-white-paper/5g-white-paper.html}
\BIBentrySTDinterwordspacing

\bibitem{osseiran2014scenarios}
A.~Osseiran, F.~Boccardi, V.~Braun, K.~Kusume, P.~Marsch, M.~Maternia,
  O.~Queseth, M.~Schellmann, H.~Schotten, H.~Taoka \emph{et~al.}, ``Scenarios
  for {5G} mobile and wireless communications: the vision of the {METIS}
  project,'' \emph{IEEE Communications Magazine}, vol.~52, no.~5, pp. 26--35,
  2014.

\bibitem{abdelwahab2016network}
S.~Abdelwahab, B.~Hamdaoui, M.~Guizani, and T.~Znati, ``Network function
  virtualization in {5G},'' \emph{IEEE Communications Magazine}, vol.~54,
  no.~4, pp. 84--91, 2016.

\bibitem{han2015network}
B.~Han, V.~Gopalakrishnan, L.~Ji, and S.~Lee, ``Network function
  virtualization: Challenges and opportunities for innovations,'' \emph{IEEE
  Communications Magazine}, vol.~53, no.~2, pp. 90--97, 2015.

\bibitem{peng2015fronthaul}
M.~Peng, C.~Wang, V.~Lau, and H.~V. Poor, ``Fronthaul-constrained cloud radio
  access networks: Insights and challenges,'' \emph{IEEE Wireless
  Communications}, vol.~22, no.~2, pp. 152--160, 2015.

\bibitem{mijumbi2016network}
R.~Mijumbi, J.~Serrat, J.-L. Gorricho, N.~Bouten, F.~De~Turck, and R.~Boutaba,
  ``Network function virtualization: State-of-the-art and research
  challenges,'' \emph{IEEE Communications Surveys \& Tutorials}, vol.~18,
  no.~1, pp. 236--262, 2016.

\bibitem{etsi20175G}
ETSI, ``Network functions virtualisation: White paper on {NFV} priorities for
  {5G},'' ETSI, Tech. Rep., 2017.

\bibitem{etsi15nfvi}
------, ``Network functions virtualisation ({NFV}): Infrastructure overview,''
  ETSI, Tech. Rep., 2015.

\bibitem{eramo2017approach}
V.~Eramo, E.~Miucci, M.~Ammar, and F.~G. Lavacca, ``An approach for service
  function chain routing and virtual function network instance migration in
  network function virtualization architectures,'' \emph{IEEE/ACM Transactions
  on Networking}, vol.~25, no.~4, pp. 2008--2025, 2017.

\bibitem{wang2017consistent}
W.~Wang, Y.~Liu, Y.~Li, H.~Song, Y.~Wang, and J.~Yuan, ``Consistent state
  updates for virtualized network function migration,'' \emph{IEEE Transactions
  on Services Computing}, 2017.

\bibitem{etsi2014nfvMO}
ETSI, ``Network functions virtualisation ({NFV}): Management and
  orchestration,'' ETSI, Tech. Rep. ETSI GS NFV-MAN 001, 2014.

\bibitem{xia2014VPool}
L.~Xia, Q.~Wu, D.~King, and H.~Yokota, ``{VNF} pool use cases,'' IETF, Tech.
  Rep. draft-king-vnfpool-mobile-use-case-02, 2014.

\bibitem{eramo2017migration}
V.~Eramo, M.~Ammar, and F.~G. Lavacca, ``Migration energy aware
  reconfigurations of virtual network function instances in {NFV}
  architectures,'' \emph{IEEE Access}, vol.~5, pp. 4927--4938, 2017.

\bibitem{hawilo2017orchestrating}
H.~Hawilo, M.~Jammal, and A.~Shami, ``Orchestrating network function
  virtualization platform: Migration or re-instantiation?'' in \emph{IEEE 6th
  International Conference on Cloud Networking (CloudNet)}.\hskip 1em plus
  0.5em minus 0.4em\relax IEEE, 2017, pp. 1--6.

\bibitem{han2017resiliency}
B.~Han, V.~Gopalakrishnan, G.~Kathirvel, and A.~Shaikh, ``On the resiliency of
  virtual network functions,'' \emph{IEEE Communications Magazine}, vol.~55,
  no.~7, pp. 152--157, 2017.

\bibitem{nobach2017statelet}
L.~Nobach, I.~Rimac, V.~Hilt, and D.~Hausheer, ``Statelet-based efficient and
  seamless {NFV} state transfer,'' \emph{IEEE Transactions on Network and
  Service Management}, vol.~14, no.~4, pp. 964--977, 2017.

\bibitem{wang2016transparent}
Y.~Wang, G.~Xie, Z.~Li, P.~He, and K.~Salamatian, ``Transparent flow migration
  for {NFV},'' in \emph{IEEE International Conference on Network Protocols
  (ICNP)}.\hskip 1em plus 0.5em minus 0.4em\relax IEEE, 2016, pp. 1--10.

\bibitem{RNDM18}
T.~Lin and Z.~Zhou, ``Robust virtual network function provisioning under random
  failures on network function enabled nodes,'' in \emph{Proc. of 10th
  International Workshop on Resilient Networks Design and Modeling (RNDM)},
  2018.

\bibitem{cheng2015enabling}
G.~Cheng, H.~Chen, H.~Hu, Z.~Wang, and J.~Lan, ``Enabling network function
  combination via service chain instantiation,'' \emph{Computer Networks},
  vol.~92, pp. 396--407, 2015.

\bibitem{jalalitabar2018service}
M.~Jalalitabar, ``Service function graph design and embedding in next
  generation internet,'' Ph.D. dissertation, Georgia State University, 2018.

\bibitem{nfv2014001}
G.~NFV-SWA, ``001-v1. 1.1-network functions virtualisation (nfv); virtual
  network functions architecture,'' 2014.

\bibitem{openSFC}
\BIBentryALTinterwordspacing
Openstack, ``Ietf sfc encapsulation,'' Feb. 2018. [Online]. Available:
  \url{https://docs.openstack.org/networking-sfc/latest/contributor/ietf_sfc_encapsulation.html}
\BIBentrySTDinterwordspacing

\bibitem{bonfim2018integrated}
M.~S. Bonfim, K.~L. Dias, and S.~F. Fernandes, ``Integrated {NFV/SDN}
  architectures: A systematic literature review,'' \emph{arXiv preprint
  arXiv:1801.01516}, 2018.

\bibitem{nguyen2017sdn}
V.-G. Nguyen, A.~Brunstrom, K.-J. Grinnemo, and J.~Taheri, ``{SDN/NFV}-based
  mobile packet core network architectures: A survey,'' \emph{IEEE
  Communications Surveys \& Tutorials}, vol.~19, no.~3, pp. 1567--1602, 2017.

\bibitem{nguyen2017resource}
N.~C. Nguyen, P.~Wang, D.~Niyato, Y.~Wen, and Z.~Han, ``Resource management in
  cloud networking using economic analysis and pricing models: A survey,''
  \emph{IEEE Communications Surveys \& Tutorials}, 2017.

\bibitem{xie2016service}
Y.~Xie, Z.~Liu, S.~Wang, and Y.~Wang, ``Service function chaining resource
  allocation: A survey,'' \emph{arXiv preprint arXiv:1608.00095}, 2016.

\bibitem{carpio2017vnf}
F.~Carpio, S.~Dhahri, and A.~Jukan, ``{VNF} placement with replication for load
  balancing in {NFV} networks,'' in \emph{IEEE International Conference on
  Communications (ICC)}.\hskip 1em plus 0.5em minus 0.4em\relax IEEE, 2017, pp.
  1--6.

\bibitem{zhang2018theory}
J.~Zhang, W.~Wu, and J.~Lui, ``On the theory of function placement and chaining
  for network function virtualization,'' in \emph{Proceedings of the Eighteenth
  ACM International Symposium on Mobile Ad Hoc Networking and Computing}.\hskip
  1em plus 0.5em minus 0.4em\relax ACM, 2018, pp. 91--100.

\bibitem{zeng2016orchestrating}
M.~Zeng, W.~Fang, and Z.~Zhu, ``Orchestrating tree-type {VNF} forwarding graphs
  in inter-{DC} elastic optical networks,'' \emph{Journal of Lightwave
  Technology}, vol.~34, no.~14, pp. 3330--3341, 2016.

\bibitem{d2017game}
S.~D'Oro, L.~Galluccio, S.~Palazzo, and G.~Schembra, ``A game theoretic
  approach for distributed resource allocation and orchestration of softwarized
  networks,'' \emph{IEEE Journal on Selected Areas in Communications}, vol.~35,
  no.~3, pp. 721--735, 2017.

\bibitem{alameddine2017interplay}
H.~A. Alameddine, S.~Sebbah, and C.~Assi, ``On the interplay between network
  function mapping and scheduling in {VNF}-based networks: A column generation
  approach,'' \emph{IEEE Transactions on Network and Service Management},
  vol.~14, no.~4, pp. 860--874, 2017.

\bibitem{sallam2018shortest}
G.~Sallam, G.~R. Gupta, B.~Li, and B.~Ji, ``Shortest path and maximum flow
  problems under service function chaining constraints,'' \emph{arXiv preprint
  arXiv:1801.05795}, 2018.

\bibitem{gupta2018scalable}
A.~Gupta, B.~Jaumard, M.~Tornatore, and B.~Mukherjee, ``A scalable approach for
  service chain mapping with multiple {SC} instances in a wide-area network,''
  \emph{IEEE Journal on Selected Areas in Communications}, vol.~36, no.~3, pp.
  529--541, 2018.

\bibitem{cohen2015near}
R.~Cohen, L.~Lewin-Eytan, J.~S. Naor, and D.~Raz, ``Near optimal placement of
  virtual network functions,'' in \emph{IEEE INFOCOM}.\hskip 1em plus 0.5em
  minus 0.4em\relax IEEE, 2015, pp. 1346--1354.

\bibitem{bari2015orchestrating}
M.~F. Bari, S.~R. Chowdhury, R.~Ahmed, and R.~Boutaba, ``On orchestrating
  virtual network functions,'' in \emph{11th International Conference on
  Network and Service Management (CNSM)}.\hskip 1em plus 0.5em minus
  0.4em\relax IEEE, 2015, pp. 50--56.

\bibitem{chowdhury2009virtual}
N.~M.~K. Chowdhury, M.~R. Rahman, and R.~Boutaba, ``Virtual network embedding
  with coordinated node and link mapping,'' in \emph{IEEE INFOCOM}.\hskip 1em
  plus 0.5em minus 0.4em\relax IEEE, 2009, pp. 783--791.

\bibitem{rost2016service}
M.~Rost and S.~Schmid, ``Service chain and virtual network embeddings:
  Approximations using randomized rounding,'' \emph{arXiv preprint
  arXiv:1604.02180}, 2016.

\bibitem{even2016approximation}
G.~Even, M.~Rost, and S.~Schmid, ``An approximation algorithm for path
  computation and function placement in sdns,'' in \emph{International
  Colloquium on Structural Information and Communication Complexity}.\hskip 1em
  plus 0.5em minus 0.4em\relax Springer, 2016, pp. 374--390.

\bibitem{lin2016demand}
T.~Lin, Z.~Zhou, M.~Tornatore, and B.~Mukherjee, ``Demand-aware network
  function placement,'' \emph{Journal of Lightwave Technology}, vol.~34,
  no.~11, pp. 2590--2600, 2016.

\bibitem{5gNormaQoS}
5G-NORMA, ``Definition and specification of connectivity and {QoE/QoS}
  management mechanisms,'' 5G NORMA, Tech. Rep. H2020-ICT-2014-2 5G NORMA,
  2017.

\bibitem{fan2015grep}
J.~Fan, Z.~Ye, C.~Guan, X.~Gao, K.~Ren, and C.~Qiao, ``{GREP}: Guaranteeing
  reliability with enhanced protection in {NFV},'' in \emph{ACM SIGCOMM
  Workshop on Hot Topics in Middleboxes and Network Function
  Virtualization}.\hskip 1em plus 0.5em minus 0.4em\relax ACM, 2015, pp.
  13--18.

\bibitem{qu2016reliability}
L.~Qu, C.~Assi, K.~Shaban, and M.~Khabbaz, ``Reliability-aware service
  provisioning in {NFV}-enabled enterprise datacenter networks,'' in
  \emph{International Conference on Network and Service Management
  (CNSM)}.\hskip 1em plus 0.5em minus 0.4em\relax IEEE, 2016, pp. 153--159.

\bibitem{guha1999greedy}
S.~Guha and S.~Khuller, ``Greedy strikes back: Improved facility location
  algorithms,'' \emph{Journal of algorithms}, vol.~31, no.~1, pp. 228--248,
  1999.

\bibitem{krishnaswamy2012inapproximability}
R.~Krishnaswamy and M.~Sviridenko, ``Inapproximability of the multi-level
  uncapacitated facility location problem,'' in \emph{Proceedings of the
  twenty-third annual ACM-SIAM symposium on Discrete algorithms}.\hskip 1em
  plus 0.5em minus 0.4em\relax Society for Industrial and Applied Mathematics,
  2012, pp. 718--734.

\bibitem{li20131}
S.~Li, ``A 1.488 approximation algorithm for the uncapacitated facility
  location problem,'' \emph{Information and Computation}, vol. 222, pp. 45--58,
  2013.

\bibitem{shmoys1997approximation}
D.~B. Shmoys, {\'E}.~Tardos, and K.~Aardal, ``Approximation algorithms for
  facility location problems,'' in \emph{Proceedings of the twenty-ninth annual
  ACM symposium on Theory of computing}.\hskip 1em plus 0.5em minus 0.4em\relax
  ACM, 1997, pp. 265--274.

\bibitem{aardal19993}
K.~Aardal, F.~A. Chudak, and D.~B. Shmoys, ``A 3-approximation algorithm for
  the k-level uncapacitated facility location problem,'' \emph{Information
  Processing Letters}, vol.~72, no. 5-6, pp. 161--167, 1999.

\bibitem{ageev2004improved}
A.~Ageev, Y.~Ye, and J.~Zhang, ``Improved combinatorial approximation
  algorithms for the k-level facility location problem,'' \emph{SIAM Journal on
  Discrete Mathematics}, vol.~18, no.~1, pp. 207--217, 2004.

\bibitem{zhang2006approximating}
J.~Zhang, ``Approximating the two-level facility location problem via a
  quasi-greedy approach,'' \emph{Mathematical Programming}, vol. 108, no.~1,
  pp. 159--176, 2006.

\bibitem{byrka2007optimal}
J.~Byrka, ``An optimal bifactor approximation algorithm for the metric
  uncapacitated facility location problem,'' in \emph{Approximation,
  Randomization, and Combinatorial Optimization. Algorithms and
  Techniques}.\hskip 1em plus 0.5em minus 0.4em\relax Princeton, NJ, USA:
  Springer, 2007, pp. 29--43.

\bibitem{byrka2016improved}
J.~Byrka, S.~Li, and B.~Rybicki, ``Improved approximation algorithm for k-level
  uncapacitated facility location problem (with penalties),'' \emph{Theory of
  Computing Systems}, vol.~58, no.~1, pp. 19--44, 2016.

\bibitem{chechik2009robust}
S.~Chechik and D.~Peleg, ``Robust fault tolerant uncapacitated facility
  location,'' \emph{arXiv preprint arXiv:0912.3188}, 2009.

\bibitem{byrka2010lp}
J.~Byrka, M.~Ghodsi, and A.~Srinivasan, ``Lp-rounding algorithms for
  facility-location problems,'' \emph{arXiv preprint arXiv:1007.3611}, 2010.

\bibitem{ETSI2016Reliable}
\BIBentryALTinterwordspacing
ETSI, ``Network functions virtualisation ({NFV}); reliability; report on models
  and features for end-to-end reliability,'' ETSI, Tech. Rep. V1.1.1, 2016.
  [Online]. Available:
  \url{http://www.etsi.org/deliver/etsi_gs/NFV-REL/001_099/003/01.01.01_60/gs_nfv-rel003v010101p.pdf}
\BIBentrySTDinterwordspacing

\bibitem{vardhan2009finding}
H.~Vardhan, S.~Billenahalli, W.~Huang, M.~Razo, A.~Sivasankaran, L.~Tang,
  P.~Monti, M.~Tacca, and A.~Fumagalli, ``Finding a simple path with multiple
  must-include nodes,'' in \emph{2009 IEEE International Symposium on Modeling,
  Analysis \& Simulation of Computer and Telecommunication Systems}.\hskip 1em
  plus 0.5em minus 0.4em\relax IEEE, 2009, pp. 1--3.

\bibitem{chechik2014robust}
S.~Chechik and D.~Peleg, ``Robust fault tolerant uncapacitated facility
  location,'' \emph{Theoretical Computer Science}, vol. 543, pp. 9--23, 2014.

\end{thebibliography}
\end{document}